\documentclass[12pt]{iopart}

\usepackage{graphicx,setspace}
\usepackage{psfrag}
\usepackage{amsfonts}
\expandafter\let\csname equation*\endcsname\relax
\expandafter\let\csname endequation*\endcsname\relax
\usepackage{amsmath}
\usepackage{amssymb, amsthm}
\usepackage{mathrsfs}
\usepackage{epstopdf}
\usepackage{float}
\usepackage{lineno,hyperref}
\usepackage{color}
\usepackage{subfigure}

\usepackage{amsmath}
\usepackage{dsfont} 
\usepackage{amssymb} 
\usepackage{graphicx,color}
\usepackage{extarrows}
\usepackage{CJK}
\usepackage{graphicx}

\usepackage{epstopdf}

\usepackage{mathrsfs} 
\usepackage{caption}
\usepackage{subfigure}

\modulolinenumbers[5]
\begin{document}
\newtheorem{definition}{Definition}[section]
\newtheorem{lemma}{Lemma}[section]
\newtheorem{remark}{Remark}[section]
\newtheorem{theorem}{Theorem}[section]
\newtheorem{proposition}{Proposition}
\newtheorem{assumption}{Assumption}
\newtheorem{example}{Example}
\newtheorem{corollary}{Corollary}[section]
\def\ep{\varepsilon}
\def\Rn{\mathbb{R}^{n}}
\def\Rm{\mathbb{R}^{m}}
\def\E{\mathbb{E}}
\def\hte{\hat\theta}
\renewcommand{\theequation}{\thesection.\arabic{equation}}

\title[Onsager-Machlup Function for Jump-Diffusions]{The Onsager-Machlup Function as Lagrangian for the Most Probable Path  of  a Jump-Diffusion Process}

\author{Ying Chao$^1$ and
Jinqiao Duan$^{2,*}$}

\address{$^1$ School of Mathematics and Statistics \& Center for Mathematical Sciences \& Hubei Key Laboratory for Engineering Modeling and Scientific Computing, Huazhong University of Science and Technology, Wuhan 430074,  China}
\address{$^{2}$Department of Applied Mathematics, Illinois Institute of Technology, Chicago, IL 60616, USA}

\address{$^*$ Author to whom correspondence should be addressed: duan@iit.edu}
\ead{yingchao1993@hust.edu.cn, duan@iit.edu.}

\begin{abstract}
This work is devoted to deriving the Onsager-Machlup function for a class of stochastic dynamical systems under  (non-Gaussian) L\'{e}vy noise  as well as (Gaussian) Brownian noise, and examining the corresponding most probable paths.
This Onsager-Machlup function is the Lagrangian giving the most probable path connecting metastable states for jump-diffusion processes. This is done by applying the Girsanov transformation for measures induced by jump-diffusion processes. Moreover, we have found this Lagrangian function is consistent with the result in the special case of diffusion processes. Finally, we apply this new Onsager-Machlup function to investigate dynamical behaviors  analytically and numerically in several
examples. These include the transitions  from one metastable state to another metastable state in a double-well system, with numerical experiments illustrating most probable transition paths for various noise parameters.

\end{abstract}

\section{Introduction}

A diffusion process is usually a solution process of a differential equation driven by Brownian motion.
The Onsager-Machlup (OM) function of a diffusion process was first introduced  in \cite{OM} and  also considered  in \cite{St, DB}, among many more recent efforts.  It appears in the  asymptotic probability estimation   of sample paths of a diffusion process lying within a tube along a   smooth path.   In fact, this asymptotic likelihood   contains an exponential term whose    exponent  is a negative integral over a   time interval, and the      integrand in this integral is the so-called OM function.  For convenience, we also call
the integral of OM function over a time interval the OM action functional.
 In other words,  the OM function helps to evaluate the asymptotic probability  measure of a small neighborhood of a smooth path. Taking the idea of D\"{u}rr and Bach \cite{DB}, another more physical meaning can be given to the OM function, which interprets the OM function as a Lagrangian for determining the most probable tube of a diffusion process by means of a variation principle. Thus, the OM function serves as a tool to study some dynamical behaviors of stochastic dynamical systems, such as most probable transition paths connecting metastable states. \par

The method concerning the derivation of the OM function for a diffusion process has been developed in the preceding years. Onsager and Machlup  \cite{OM} were the first to derive the OM function for diffusion processes with linear drift and constant diffusion coefficients.  Subsequently, Tisza and Manning \cite{TM} began to generalize the results of \cite{OM} to nonlinear equations. Takahashi-Watanabe \cite{TS} and Fujita-Kotani \cite{FK} further derived  OM function for Brownian motion on a complete Riemannian manifold. Another purely probabilistic proof of Onsager-Machlup formula for diffusion processes on a Riemannian manifold was given by Hara and Takahashi \cite{HT}. There are other works devoted to deriving the OM functions for   diffusion processes using various approaches \cite{St, BDS, DB}. For example, D\"{u}rr and Bach \cite{DB} derived the OM function based on the Girsanov transformation for measures induced by diffusion processes (see also \cite{IS}). In addition, an application concerned with the fluctuation phenomena is described by Singh \cite{S} and Deza {\em et al} \cite{DGW}. More recently, OM function has been used to analyze issues related to the trajectory entropy of the overdamped Langevin equation \cite{HYC} as well as to data assimilation \cite{Sug, B}. Note that  most existing works mentioned above are  for diffusion processes, i.e. solution processes of stochastic differential equations with (Gaussian)  Brownian motion.\par

However, random fluctuations in complex biological and physical systems are often non-Gaussian rather than Gaussian \cite{ZLDK, WCDKL, W}.  L\'{e}vy motions are appropriate models for a class of important non-Gaussian processes with jumps or bursts \cite{ BSW}. In the sense of L\'evy-It\^o decomposition, a general L\'evy motion could be understood as a sum of a Brownian motion and a pure jump L\'evy motion, in addition to a  drift term which may be absorbed in the vector field in stochastic differential equation (SDE) models for these complex systems. Note that Ditlevsen \cite{DPD} found that a climate change system may be modelled by an SDE with Brownian motion and L\'{e}vy motion, and an objective is to understand climate transitions between climate metastable states. An SDE with Brownian motion and L\'{e}vy motion also appeared for modeling certain gene transcriptions in biology \cite{ZLDK}, with a goal to quantify the transitions between low and high protein concentrations.
Motivated by these modelling efforts,  we  plan to investigate    transition phenomena for these non-Gaussian complex systems.  Especially,  it is desirable to consider the  OM functions for jump-diffusion processes, defined as solution processes of stochastic differential equations with Brownian motion $B_t$ and L\'evy motion  $L_t$, and examine the corresponding most probable paths between metastable paths.  As far as we know, this is not yet done. 


We remark that Moret and Nualart \cite{MN}  recently derived  the OM function for the fractional (but still Gaussian) Brownian motion.  Bardina {\em et al} \cite{XCS} dealt with the asymptotic evaluation of the Poisson measure for a tube around a jump path  but that work is not for a general L\'evy motion.


In this present paper, we consider a class of one-dimensional nonlinear stochastic dynamical systems   with     Brownian noise and   L\'{e}vy noise. The solution process is a jump-diffusion process.    Using the Girsanov transformation for a measure induced by this jump-diffusion process, we derive the Onsager-Machlup function, which  may still  be regarded  as the Lagrangian giving the most probable tube around a smooth path as in the case of diffusion processes \cite{DB}. Different from the case of diffusion processes,  an extra term depicting the impact of  L\'{e}vy noise will appear in  the  OM function. It is worthwhile to mention that our result is consistent with the result of the diffusion processes     \cite{OM, DB,BDS} in the absence of  L\'{e}vy noise. Thus,  our work extracts the non-trivial effect of  pure jump   L\'{e}vy noise in stochastic dynamics and we will see that this will affect the most probable paths.   With the help of OM function and a variation principle, the most probable path from one     state to  another, especially the most probable path from one metastable state to another metastable state over a finite time interval,  for some stochastic dynamical systems can be   found by numerical simulation.  We will illustrate this for a stochastic system with linear potential and a prototypical double-well system under L\'evy noise and Brownian noise.


An inspiration for this paper goes back to the work \cite{DB} by D\"urr and Bach. We generalize and improve their results to  a class of dynamical systems under random fluctuations consisting of  Brownian noise and  L\'evy noise. This will give rise to several difficulties both in analytic and probabilistic aspects as the non-Gaussian pure jump L\'evy noise is present. In addition, we use  the same notations as   in \cite{DB} to compare our results with the diffusion case.\par

This paper is organized as follows.  In Section 2, we recall some basic concepts about L\'{e}vy motions,  and present the framework  for deriving the OM function of a stochastic dynamical system.  After setting up the theory of induced measures in Section 3,  we  derive the OM function   in Section 4.   Subsequently in Section 5,  we present the equation of motion for the most probable path of a class of jump-diffusion processes    and discuss its solutions.  Some concrete examples are tested to illustrate our results in Section 6. Finally, we summarize our work in Section 7.

\section{Preliminaries}
We now  recall some basic facts about L\'evy motions \cite{Ap, Sato, Duan}, introduce a class of  stochastic differential equations to be studied, and define the OM function and OM functional.
\subsection{L\'{e}vy motions}
Let $(\Omega, \mathcal{F}, (\mathcal{F}_t)_{t\geq0}, \mathds{P})$ be a filtered probability space, where $\mathcal{F}_t$ is a nondecreasing family of sub-$\sigma$-fields of $\mathcal{F}$ satisfying the usual conditions. An $\mathcal{F}_t$-adapted stochastic process $L_t=L(t)$ taking values in $\mathbb{R}^n$ with $L(0)=0$ $a.s.$ is called a L\'{e}vy motion if it is stochastically continuous,   with independent increments and stationary increments.\par

An $n$-dimensional L\'{e}vy motion can be characterized by a drift vector $b\in \mathbb{R}^n$, an $n \times n$ non-negative-definite, symmetric matrix Q,  and a Borel measure $\nu$ defined on ${\mathbb{R}^n}\backslash \{ 0\}$. We call $(b,Q,\nu)$ the generating triplet of the L\'{e}vy motion $L_t$. Moreover, we have the following L\'{e}vy-It\^o decomposition for $L_t$:
\begin{equation}\label{decomposition}
{L_t}= bt + B_{Q}(t) + \int_{|x|< 1} x \tilde N(t,dx) + \int_{|x|\ge 1} x N(t,dx),\notag
\end{equation}
where $N(dt,dx)$ is the Poisson random measure on $\mathbb{R}^{+}\times({\mathbb{R}^n}\backslash \{ 0\})$, $\tilde N(dt,dx) = N(dt,dx) - \nu (dx)dt$ is the compensated Poisson random measure, $\nu(A)=EN(1,A)$ with $A\in\mathcal{B}({\mathbb{R}^n}\backslash \{ 0\})$ is the jump measure, and $B_Q(t)$ is an independent $n$-dimensional Brownian motion with covariance matrix $Q$. Here $|\cdot|$ denotes the Euclidean norm.\par

A scalar $\alpha$-stable L\'evy motion  $L_t^{\alpha, \beta}$ is a special L\'evy motion, with
 non-Gaussianity index (or  stability index) $\alpha\in(0,2)$ and skewness parameter $\beta\in[-1,1]$. It has the generating triplet $(b,0,\nu_{\alpha, \beta})$. Its jump measure $\nu_{\alpha, \beta}$ is in the form of $$\nu_{\alpha, \beta} (d\xi)=c_1|\xi|^{-1-\alpha}\chi_{(0,\infty)}(\xi)d\xi+c_2|\xi|^{-1-\alpha}\chi_{(-\infty,0)}(\xi)d\xi$$
with $\beta=\frac{c_1-c_2}{c_1+c_2}$, $c_1=k_{\alpha}\frac{1+\beta}{2}$, and $c_2=k_{\alpha}\frac{1-\beta}{2}$, where
$$
k_{\alpha}=
\left \{
  \begin{array}{ll}
    \frac{\alpha(1-\alpha)}{\Gamma(2-\alpha)\cos(\frac{\pi\alpha}{2})}, &  \hbox{if $\alpha<1$,} \\
    \frac{2}{\pi}, & \hbox{if $\alpha=1$.}
  \end{array}
\right.
$$
Here, $\chi(\xi)$ denotes indicator function and $\Gamma$ is the Gamma function. If drift vector $b=0$, and $\nu_{\alpha, \beta}$ is symmetric, i.e. $\beta=0$, then this stable process is called a symmetric $\alpha$-stable process. Note that $\int_{|\xi|<1}|\xi|\nu_{\alpha, \beta}(d\xi)$ is finite if and only if $\alpha<1$ and the integral $\int_{|\xi|\geq1}|\xi|\nu_{\alpha, \beta}(d\xi)$ is finite if and only if $\alpha>1$.

\subsection{ Framework }
Consider the following scalar stochastic differential equation (SDE) defined on $\tau=[s,u]$ with respect to $(\Omega, \mathcal{F}, (\mathcal{F}_t)_{t\geq0}, \mathds{P})$
\begin{equation} \label{Equation-1}
  \begin{split}
  &dX_t=f(X_{t-})dt+g(X_{t-})dB_t+dL_t,    \\
  &X_s=x_0\in \mathbb{R}.
  \end{split}
\end{equation}

\noindent Where ${L_t}$ is a L\'{e}vy process with characteristics $(0,0,\nu)$ satisfying $\int_{|\xi|<1}\xi\nu(d\xi)<\infty$, $B_t$ is a standard Brownian motion and $X_{t-}$ is the left limit at the point $t$, i.e. $X_{t-}=\lim_{s\uparrow t}X_{s}$.\par
More precisely, the above equation can be written in the following form by employing the L\'{e}vy-It\^o decomposition :
\begin{equation} \label{Equation-2}
  \begin{split}
  &dX_t=f(X_{t-})dt+g(X_{t-})dB_t+\int_{|x|<1}x\tilde{N}(dt,dx)+\int_{|x|\geq1}xN(dt,dx),    \\
  &X_s=x_0\in \mathbb{R}.\notag
  \end{split}
\end{equation}
\noindent The last term in the preceding equation involves large jumps that are controlled by $x$. This term can be handled by using interlacing \cite[ Page 365] {Ap}, and it makes sense to begin by omitting this term and concentrate on the study of the equation driven by continuous noise interspersed with small jumps. To this end, we consider the following SDE:
\begin{equation} \label{Equation-main}
  \begin{split}
  &dX_t=f(X_{t-})dt+g(X_{t-})dB_t+\int_{|x|<1}x\tilde{N}(dt,dx),   \\
  &X_s=x_0\in \mathbb{R}.
  \end{split}
\end{equation}

\noindent If we further assume that $f(x)$ and $g(x)$ are locally Lipschitz continuous functions and satisfy ``one sided linear growth'' condition in the following sense:\par
\noindent{\bf C1 }(Locally Lipschitz condition) For any $R>0$, there exists $K_1>0$ such that, for all $|y_1|, |y_2|\leq R$,
$$|f(y_1)-f(y_2)|^2+|g(y_1)-g(y_2)|^2\leq K_1|y_1-y_2|^2,$$
{\bf C2 }(One sided linear growth condition) There exists $K_2>0$ such that, for all $y\in\mathbb{R}$,
$$|g(y)|^2+2y\cdot f(y)\leq K_2(1+|y|^2),$$
then there exists a unique global solution to (\ref{Equation-main}) and the solution process is adapted and c\`{a}dl\`{a}g, refer to Theorem 3.1 of \cite{ABW}, or \cite{XZ}. We call a solution of (\ref{Equation-main}) a jump-diffusion process. As we see, it happens to be a diffusion process only in the absence of L\'{e}vy noise. And note that there is by no means universal agreement about the use of the phrase jump-diffusion process. Thus, it makes sense to use this terminology to denote the solution of (\ref{Equation-main}) throughout the paper  \cite[ Page 383] {Ap}.\par




Now we will introduce some notations and terminologies to be used throughout the paper.
Denote by $C^1(\mathbb{R})$ the set of all real-valued continuous differentiable functions, $C^2(\mathbb{R})$ the set of all real-valued twice continuous differentiable functions and $ H^1([s, u];\mathbb{R})$ the set of all real-valued square integrable functions with square integrable weak derivative defined on $\tau=[s,u]$.
Denote the space of paths of a solution process of  (\ref{Equation-main}) by $D_\tau^{x_0}$, which is the space of c\`{a}dl\`{a}g functions, i.e., $$D_\tau^{x_0}=\{x(t)\mid x:\tau\rightarrow \mathbb{R}, x(t) \hbox{ continuous from the right and have limits on the left}, x(s)=x_0\}.$$
\noindent For each finite subset $S$ of $[s, u]$ write $\pi_S$ for the projection map from $D_\tau^{x_0}$ into $\mathbb{R}^S$ that takes an $x$ onto its vector of values $\{x(t): t\in S\}$. These projections generate the projection $\sigma$-field $\mathcal{P}$.  The  measurability in this paper will always refer to this $\sigma$-field.\par
As every  c\`{a}dl\`{a}g function on $\tau$ is bounded, equipping $D_\tau^{x_0}$ with the uniform norm $\|\cdot\|$
$$\|x\|=\sup_{t\in[s,u]}|x(t)|, x(t)\in D_\tau^{x_0},$$
$D_\tau^{x_0}$ is a Banach space. No other norm will be used for $D_\tau^{x_0}$ in this paper. We could show that the $\sigma$-field $\mathcal{B}_0$ generated by closed balls for uniform norm is the projection $\sigma$-field $\mathcal{P}$,  but not the larger Borel $\sigma$-field $\mathcal{B}_\tau^{x_0}$. For more details, see \cite[ Page 87-90] {Pollard}. \par
\begin{lemma}\label{lemma 2.1} $\mathcal{B}_0=\mathcal{P}$ and $\mathcal{P}\subsetneqq\mathcal{B}_\tau^{x_0}.$
\end{lemma}
\begin{proof}Consider a closed ball $B_\rho=\{y\in D_\tau^{x_0}:\sup_{t\in[s,u]}|y_t-z_t|\leq\rho\}$, where $z$ is an element of $D_\tau^{x_0}$. Since the functions in $D_\tau^{x_0}$ are right-continuous,
$$B_\rho={\cap}_{t_k} \{y\in D_\tau^{x_0}:|y_{t_k}-z_{t_k}|\leq\rho\}\in \mathcal{P},$$
where $t_k$ are the rational points of $[s,u]$. Therefore $\mathcal{B}_0\subset\mathcal{P}$. Conversely, we can express $B=\{x: x_{t}>a\}$ as a countable union of closed balls. Thus $\mathcal{B}_0$=$\mathcal{P}$. \par
Consider a cylinder set $B=\{x: x_{t_1}<b\}$. For every $x\in B$, we have $x_{t_1}\leq b-\epsilon$ where $\epsilon$ can be chosen sufficiently small. We find that there exists a neighborhood of $x$, denoted by $K(x,\epsilon)=\{y\in D_\tau^{x_0}\mid \|x-y\|<\epsilon,\epsilon>0\}$. For every $y\in K(x,\epsilon)$, $|y_{t_1}-x_{t_1}|\leq\parallel y-x\parallel<\epsilon$, i.e. $y_{t_1}<x_{t_1}+\epsilon\leq b$, thus $K(x,\epsilon)\subset B$ and $B$ is an open set. Hence it follows that $\{x: x_{t_1}<b_1,...,x_{t_n}<b_n\}\in \mathcal{B}_\tau^{x_0}$, and thus $\mathcal{P}\subset\mathcal{B}_\tau^{x_0}$. Since $D_\tau^{x_0}$ equipped with the uniform norm is not separable, open subset of a metric space  $D_\tau^{x_0}$ can not be written as a countable union of closed balls. Thus, $\mathcal{B}_\tau^{x_0}\nsubseteq \mathcal{B}_0$.\par
The proof is complete.
\end{proof}
Since this paper is interested in the problem of finding the most probable tube of $X_t$, it only makes sense to ask for the probability that paths lie within the closed tube $K(z,\epsilon)$:
\begin{equation}\label{tube}
K(z,\epsilon)=\{x\in D_\tau^{x_0}\mid z\in D_\tau^{x_0},\|x-z\|\leq\epsilon,\epsilon>0\}.
\end{equation}
In fact, this probability can be calculated or estimated by using the induced measure in function space.\par
Define the measure $\mu_X$ on $\mathcal{P}$ induced by the process (\ref{Equation-main}) via
$$\mu_X(B)=\mathds{P}(\{\omega\in\Omega\mid X_t(\omega)\in B\}), \;\;  \mbox{ for }   B\in \mathcal{P}.$$
Thus, by means of the induced measure, once an $\epsilon>0$ is given we can compare the probabilities of tubes of the same ``thickness'' for all $z\in D_\tau^{x_0}$ using
\begin{equation}\label{tubeP}
\mu_X(K(z,\epsilon))=\mathds{P}(\{\omega\in\Omega\mid X_t(\omega)\in K(z,\epsilon)\}),
\end{equation}
as $K(z,\epsilon)\in \mathcal{P}$. Note that the computation of probability $\mu_X(K(z,\epsilon))$ is an object for investigation in the next section.\par
\begin{remark}\label{Remark 2.1}
We choose the uniform norm   since we are able to prove that the closed balls for uniform norm generate the same $\sigma$-field $\mathcal{P}$. In fact, to ensure the measurability of a closed tube, we only need to prove $\mathcal{B}_0\subset\mathcal{P}$. Note that $D_\tau^{x_0}$ equipped with Skorohod norm is separable and thus the projection $\sigma$-field $\mathcal{P}$ is the larger Borel $\sigma$-field $\mathcal{B}_\tau^{x_0}$. But, it is of little help in our case, since Skorohod norm can not be handled easily in the following derivation.
\end{remark}
\renewcommand{\theequation}{\thesection.\arabic{equation}}
\setcounter{equation}{0}

\subsection{Definition of the Onsager-Machulup Function}
As mentioned above,  we are concerned with the problem of finding the most probable tube $K(z,\epsilon)$ given by (\ref{tube}).
As this tube depends on the choice of a function $z(t)$, we have to look for that function $z(t)$ which maximizes (\ref{tubeP}). If we restrict ourselves on differentiable functions $z(t)$, then the following definition makes sense \cite{OM, TM}.\par
\begin{definition}\label{Definition 2.1} Let $\epsilon>0$ be given. Consider a tube surrounding a reference path $z(t)$. If for $\epsilon$ sufficiently small we estimate the probability of the solution process $X_t$ lying in this tube in the form :
$$\mathds{P}(\{\|X-z\|\leq\epsilon\})\propto C(\epsilon)\exp\{-\frac{1}{2}\int_s^uOM(\dot{z},z)dt\},$$
then integrand $OM(\dot{z},z)$ is called Onsager-Machulup function. Where $\propto$ denotes the equivalence relation for $\epsilon$ small enough. We also call  $\int_s^u OM(\dot{z},z)dt$ the  Onsager-Machulup  functional. In analogy to classical mechanics, we also call the OM function the Lagrangian function, and the OM functional the action functional.
\end{definition}

\begin{remark}\label{remark 2.2}
 Br\"{o}cker \cite{B} has recently demonstrated that minimising paths of the OM functional are more typical for the dynamics and in fact carry a rigorous interpretation as most probable path of diffusion processes. Thus, inspirited by this work \cite{B}, it is reasonable to find most probable paths for jump-diffusion processes by minimising the Onsager-Machulup  functional rather than any other functional. In particular, for an SDE with pure jump L\'evy noise, Definition \ref{Definition 2.1} would be still applicable, and the minimizer of  the Onsager-Machulup  functional $\int_s^u OM(\dot{z},z)dt$  gives a notion of most probable path for this system. In addition, the minimizer  $z(t)$ may be chosen from a more general function space.
\end{remark}
If we restrict ourselves on twice differentiable functions $z(t)$, another more physical meaning can be given to the OM function, which interprets the OM function as a Lagrangian for determining the most probable tube of a jump-diffusion process. In other words, let $z_m(t)$ be a path  that maximizes $\mu_X(K(z,\epsilon))$. For $\epsilon$ sufficiently small, the path  $z_m(t)$ can be found by variation of the OM action  functional $\int_s^u OM(\dot{z},z)dt$.
This idea comes from \cite{DB} for diffusion processes and it is applicable to our setting.


\renewcommand{\theequation}{\thesection.\arabic{equation}}
\setcounter{equation}{0}

\section{Reformulation of the Radon-Nikodym derivative of induced measures}
In this section, we establish some facts about the absolute continuity between induced measure and the quasi-translation invariant measure.

\subsection{The absolute continuity between induced measures}

Recall that $\mu_X$ is absolutely continuous with respect to $\mu_Y$ and write $\mu_X\ll\mu_Y$ if $\mu_Y(A)=0$, then $\mu_X(A)=0$ for all $A\in\mathcal{P}$. And we will call measures $\mu_X, \mu_Y$ equivalent ($\mu_X\sim\mu_Y$) if $\mu_X$ is absolutely continuous with respect to $\mu_Y$ ($\mu_X\ll\mu_Y$) and if $\mu_Y\ll\mu_X$, refer to (\cite[Page 161] {Oksendal}, \cite{Sato}).\par

\begin{lemma}\label{lemma 3.1} Let $X_t$ and $Y_t$ be two jump-diffusion processes defined by the SDEs with respect to $(\Omega, \mathcal{F}, (\mathcal{F}_t)_{t\geq0}, \mathds{P})$
\begin{eqnarray}\label{Equation-x}
dX_t=f(X_{t-})dt+g(X_{t-})dB_t+\int_{|x|<1}x\tilde{N}(dt,dx),
\end{eqnarray}
\begin{eqnarray}\label{Equation-y}
dY_t=k(Y_{t-})dt+g(Y_{t-})dB_t+\int_{|x|<1}x\tilde{N}(dt,dx),
\end{eqnarray}
where the driven L\'{e}vy process with characteristic triplet $(0,0,\nu)$ satisfies $\int_{|\xi|<1}\xi\nu(d\xi)<\infty$. Here, $X_s=Y_s=x_0\in \mathbb{R}$, $t\in\tau$; $f,k,g\in C^1(\mathbb{R})$ satisfy condition ${\bf C2}$. 
Then we have $\mu_X\sim\mu_Y$ and the Radon-Nikodym derivative of $\mu_X$ with respect to $\mu_Y$ is given by
\begin{equation}\label{RND1}
\frac{d\mu_X}{d\mu_Y}[Y_t(\omega)]=\exp\{\int_s^u a(Y_{t-})dB_t-\frac{1}{2}\int_s^u(a(Y_{t-}))^2dt\},
\end{equation}
where
\begin{equation}\label{a}
a(x)=\frac{f(x)-k(x)}{g(x)}.
\end{equation}
\end{lemma}
\begin{proof}Introduce
 $$M_t=\exp\{\int_s^t a(Y_{r^-})dB_r-\frac{1}{2}\int_s^t a^2(Y_{r^-})dr\};   \;\;\;t\leq u$$
 and
$$dQ(\omega)=M_u(\omega)d\mathds{P}(\omega). $$
According to Girsanov transformation \cite[Page 23 Theorem 1.4] {Ishikawa}, $Q$ is a probability measure, the process $\hat{B_t}:=-\int_s^t a(Y_{r^-})dr+B_t$ is a Brownian motion with respect to $Q$, $\tilde{N}(dt,dx)$ is still the compensated Poisson random measure  with jump measure $\nu$ with respect to $Q$, and in terms of $\hat{B_t}$ the process $Y_t$ has the stochastic integral representation
\begin{equation}\label{Equation-yQ}
dY_t=f(Y_{t-})dt+g(Y_{t-})d\hat{B_t}+\int_{|x|<1}x\tilde{N}(dt,dx).
\end{equation}
Thus equation (\ref{Equation-yQ}) is the SDE of $Y_t$ with respect to $Q$. According to the uniqueness in distribution \cite [Page 410] {Ap}, we have the equality of $\mu_X$ and $\mu_Y^Q$ induced by $(Y_t,Q)$ on $\mathcal{P}$, i.e., $\mu_X=\mu_Y^Q$.\par
\noindent For every $B\in\mathcal{P}$, we have:
\begin{align}
  \mathds{P}(\{\omega\mid X_t(\omega)\in B\})
&=\mu_X(B)=\mu_Y^Q(B)=Q(\{\omega\mid Y_t(\omega)\in B\})\notag\\
&=\int_{\{\omega: Y_t(\omega)\in B\}}\frac{dQ}{d\mathds{P}}(\omega)d\mathds{P}(\omega)=\int_B\frac{dQ}{d\mathds{P}}(\omega)d\mu_Y(Y_t(\omega))\textcolor[rgb]{1,0,0}{,}
\end{align}

\noindent because $Q$ is by definition absolutely continuous with respect to $P$. So
$$\frac{d\mu_X}{d\mu_Y}[Y_t(\omega)]=\frac{dQ}{d\mathds{P}}(\omega)=M_u=\exp\{\int_s^u a(Y_{t-}(\omega))dB_t(\omega)-\frac{1}{2}\int_s^u(a(Y_{t-}(\omega)))^2dt\}.$$
 The proof is complete.
 \end{proof}

Now we express the Radon-Nikodym derivative given by (\ref{RND1}) in terms of a path integral. A key point is to transform the stochastic integral in (\ref{RND1}) using the It\^o formula \cite[ Page 251-255] {Ap}. For simplicity, consider $g(x)=c$.\par

Setting potential function
\begin{equation}\label{Va}
V(x)=\frac{1}{c}\int^xa(y)dy.
\end{equation}
Notice that jump measure $\nu$ satisfies $\int_{|\xi|<1}\xi\nu(d\xi)<\infty$ under our assumptions.\par
By using It\^o formula for stochastic integrals, we get an expression $F[y(t)]$ for (\ref{RND1}) (see Appendix), which illuminates the functional property of the Radon-Nikodym derivative on $D_\tau^{x_0}$ ($y(t)\in D_\tau^{x_0}$)
\begin{align}\label{F}
F[y(t)]=&\frac{d\mu_X}{d\mu_Y}[y(t)]\notag\\
=&\exp\{V(y(u))-V(x_0)-\frac{1}{2}\int_s^ub(y(t-))dt+\int_s^u\int_{|\xi|<1}\frac{\xi}{c}a(y(t-))\nu(d\xi)dt,
\notag\\
&-\Sigma_{s\leq t\leq u}[V(y(t))-V(y(t-))]\chi_{|\xi|<1}(\Delta y(t))\}
\end{align}
where
\begin{eqnarray}\label{b}
b(y(t-))=\{a(y(t-))\}^2+\frac{2}{c}a(y(t-))k(y(t-))+c\frac{da(x)}{dx}(y(t-)).
\end{eqnarray}\par

\begin{remark}\label{remark 3.1}
In the proof of Lemma \ref{lemma 3.1}, one trick has been used: By the Girsanov transformation, the $B_t$ term will be a new Brownian term while the $\tilde N(dt,dx)$ term remains in the sense of new measure due to the independence of $B_t$ and $N(t,\cdot)$. However, the mutual absolute continuity for measures induced by SDEs driven by a pure jump L\'evy motion could not be established, since the uniqueness in distribution cannot be obtained via Girsanov transform; see also \cite{XCS}. We should also note that in order to express  stochastic integral terms with respect to $N(dt,dx)$ into a path integral based on the definition of Poisson random measure, we make an important hypothesis: the jump measure $\nu$ satisfies $\int_{|\xi|<1}\xi\nu(d\xi)<\infty$ (see Appendix).
\end{remark}

\subsection{Quasi-translation invariant measure}

Recall that a measure $L$ is translation invariant on $\mathbb{R}^n$ if for each $E\in\mathcal{B}^n$, then $\mathcal{T}E\in\mathcal{B}^n$ and $L(\mathcal{T}E)=L(E)$ where $\mathcal{T}$ is a translation on $\mathbb{R}^n$ and $\mathcal{B}^n$ denotes the borel-$\sigma$-field on $\mathbb{R}^n$. As suggested in \cite{KHH, DB}, such a translation invariant measure $\mu_X$ does not exist in function space $D_\tau^{x_0}$. In fact, a weaker concept, i.e. quasi-translation invariant measure is adequate for our setting which is defined as follows \cite{DB} :\par
\begin{definition}\label{definition 3.1} Let $\mathcal{T}$ be a transformation $\mathcal{T}: D_\tau^{x_0}\rightarrow D_\tau^{x_0}$ such that
\begin{eqnarray}\label{t1}
\mathcal{T}x\rightarrow x+z_0,
\end{eqnarray}
where $z_0\in D_\tau^0$, $z_0$ is differentiable. $\mathcal{T}$ is measurable on $\mathcal{M}=\{A\in \mathcal{P}: \mathcal{T}^{-1}A\in \mathcal{P}\}$. Consider the jump-diffusion process $X_t$ and
\begin{eqnarray}\label{t2}
\mathcal{T}X_t=X_t+z_0(t).
\end{eqnarray}
If the induced measure $\mu_X\sim\mu_{\mathcal{T}X}$ on $\mathcal{M}$, $\mu_X$ and $\mu_{\mathcal{T}X}$ will be called quasi-translation invariant.
\end{definition}

\begin{remark}\label{remark 3.2}
The translation $\mathcal{T}$ on $D_\tau^{x_0}$ may not be measurable due to the fact that the projection $\sigma$-field $\mathcal{P}$ is not the larger Borel $\sigma$-field. To guarantee the measurability, we consider a translation restricted to a subset of $\mathcal{P}$, i.e. $\mathcal{M}$. Therefore, the quasi-translation invariant measure is well-defined. In fact, $\mathcal{M}$ is not empty. For example, if the translation $\mathcal{T}$ given by (\ref{t1}) is defined by $z_0(t)$, then we have $\mathcal{T}^{-1}(K(z,\epsilon))=K(x_0,\epsilon)$ which belongs to $\mathcal{P}$.
\end{remark}

\begin{lemma}\label{lemma 3.2} The stochastic differential equation:
\begin{equation}\label{Equation-xm1}
dX_t=f(X_{t-})dt+c\;dB_t+\int_{|x|<1}x\tilde{N}(dt,dx)
\end{equation}
  induces  a quasi-translation invariant measure $\mu_X$ and
\begin{equation}\label{RND2}
\frac{d\mu_{\mathcal{T}X}}{d\mu_X}[X_t(\omega)]=\exp\{\int_s^u a_X(X_{t-},z_0(t))dB_t-\frac{1}{2}\int_s^u(a_X(X_{t-},z_0(t)))^2dt\}
\end{equation}
where
\begin{equation}\label{aX}
a_X(x,z_0)=\frac{f(x-z_0)+\dot{z_0}-f(x)}{c}.
\end{equation}
\end{lemma}
\begin{proof}
Take a translation (\ref{t1}). By combining equation (\ref{Equation-xm1}) and translation (\ref{t2}), we get:
\begin{eqnarray}\label{TX}
d\mathcal{T}X_t&=&\{f(X_{t-})+\dot{z}_0(t)\}dt+c\;dB_t+\int_{|x|<1}x\tilde{N}(dt,dx)\notag\\
&=&\{f(\mathcal{T}X_{t-}-z_0(t))+\dot{z}_0(t)\}dt+c\;dB_t+\int_{|x|<1}x\tilde{N}(dt,dx).
\end{eqnarray}
As the diffusion and jump have not changed under $T$, we have by Lemma \ref{lemma 3.1} $\mu_{\mathcal{T}X}\sim\mu_X$ on $\mathcal{M}$ and the Radon-Nikodym derivative of $\mu_{\mathcal{T}X}$
with respect to $\mu_X$ is given by (\ref{RND2}). The proof is complete.
\end{proof}
\begin{remark}\label{remark 3.3}
The index X indicates that the term belongs to the Radon-Nikodym derivative of a measure $\mu_{\mathcal{T}X}$ with respect to $\mu_X$; if we replace $z_0$ by $-z_0$ in these functionals, they refer to the Radon-Nikodym derivative of $\mu_{\mathcal{T}^{-1}X}$ with respect to $\mu_X$. If $g(x)$ in (\ref{Equation-main}) is not constant, then Lemma \ref{lemma 3.2} is not valid. In other words, the induced measure of (\ref{Equation-main}) with $g(x)\neq c$ is not a quasi-translation invariant measure. Thus, in Section 4, we only compute the OM function for (\ref{Equation-main}) with $g(x)=c$.
\end{remark}
Taking the same procedure as above, to eliminate the stochastic integral appeared in (\ref{RND2}), we have to take into account that $V_X$, given by (\ref{VX}), is a function explicitly depending on time,
\begin{equation}\label{VX}
V_X(x,z_0)=\frac{1}{c}\int^xa_X(y,z_0(t))dy.
\end{equation}
By using It\^o formula, we get the following expression for $J_X[x(t),z_0(t)]$ (see Appendix)
\begin{align}\label{F2}
J_X[x(t),z_0(t)]
=&\frac{d\mu_{\mathcal{T}X}}{d\mu_X}[x(t)]\notag\\
=&\exp\{V_X(x(u),z_0(u))-V_X(x_0,z_0(s))-\frac{1}{2}\int_s^ud_X(x(t-),z_0(t))dt\notag\\
&+\int_s^u\int_{|\xi|<1}\frac{\xi}{c}a_X(x(t-),z_0(t))\nu(d\xi)dt\notag\\
&-\Sigma_{s\leq t \leq u}[V_X(x(t),z_0(t))-V_X(x(t-),z_0(t))]\chi_{|\xi|<1}(\Delta x(t))\},
\end{align}
where
\begin{align}\label{dX}
d_X(x(t-),z_0(t))
=&2\frac{\partial V_X(x,z_0(t))}{\partial t}\mid_{x=x(t-)}+\frac{2}{c}a_X(x(t-),z_0(t))f(x(t-))\notag\\
&+c\frac{\partial a_X(x,z_0(t))}{\partial t}\mid_{x=x(t-)}+(a_X(x(t-),z_0(t)))^2.
\end{align}

The following Lemma will show the importance of the quasi-translation invariant measures introduced above.
\begin{lemma}\label{lemma 3.3} If we take the translation (\ref{t1}) and write $\mathcal{M}$ for $\{A\in \mathcal{P}: \mathcal{T}^{-1}A\in \mathcal{P}\}$, then $\mathcal{T}$ is measurable on $\mathcal{M}$. If $\Phi(x)$ is a measurable functional on $D_\tau^{x_0}$ and $\mu_X$ is a quasi-translation invariant measure, the following equation holds on $\mathcal{M}$
\begin{eqnarray}\label{y}
\int_{A}\Phi(y)d\mu_X(y)=\int_{\mathcal{T}^{-1}A}\Phi(x+z_0)J_X[x,-z_0]d\mu_X(x).
\end{eqnarray}
\end{lemma}
\begin{proof} Similar to the proof by \cite{DB}. 
Based on the definition of $J_X[x,-z_0]$ as Radon-Nikodym derivative of $\mu_{\mathcal{T}^{-1}X}$ with respect to $\mu_X$, we have
\begin{eqnarray}
J_X[x,-z_0]d\mu_X(x)=d\mu_{\mathcal{T}^{-1}X}(x).
\end{eqnarray}
Now for every $A\in \mathcal{M}$,
\begin{align}
\mu_{\mathcal{T}^{-1}X}(\mathcal{T}^{-1}A)=P(\{\omega:\mathcal{T}^{-1}X_t(\omega)\in \mathcal{T}^{-1}A\})=P(\{\omega:X_t(\omega)\in A\}=\mu_X(A)
\end{align}
holds, which yields (\ref{y}).
\end{proof}

\renewcommand{\theequation}{\thesection.\arabic{equation}}
\setcounter{equation}{0}

\section{The Onsager-Machlup function for a jump-diffusion process}
In the preceding section, we have set up the fundamental tools needed in this section. Now let us calculate the OM function for the following stochastic differential equation
\begin{equation}\label{Equation-xmm}
dX_t=f(X_{t-})dt+c\;dB_t+\int_{|x|<1}x\tilde{N}(dt,dx),X_s=x_0\in\mathbb{R},
\end{equation}
which is a particular case of system (\ref{Equation-main})  for $g(x)=c$. We assume that $f\in C^1(\mathbb{R})$ and satisfies condition {\bf C2}. \par

Our main result about the expression of the OM function for a jump-diffusion process is present in the following theorem.\par

\begin{theorem}\label{theorem 4.1}For a class of stochastic systems in the form of (\ref{Equation-xmm}) with the jump measure satisfying $\int_{|\xi|<1}\xi \nu(d\xi)<\infty$, the Onsager-Machlup function is given, up to an additive constant, by:
\begin{eqnarray}\label{OM}
OM(\dot{z},z)=(\frac{\dot{z}-f(z)}{c})^2+f'(z)+2\frac{\dot{z}-f(z)}{c^2}\int_{|\xi|<1}\xi \nu(d\xi),
\end{eqnarray}
where $z(t)\in D_\tau^{x_0}$ is a differentiable function.
The contribution of pure jump L\'evy noise to the OM function is the  third term.
When the jump measure is absent, we cover the OM function for the case of diffusions.
\\
In terms of OM function, the measure of tube $K(z,\epsilon)$ defined as $\{x\in D_\tau^{x_0}\mid z\in D_\tau^{x_0},\|x-z\|\leq\epsilon,\epsilon>0\}$ can be approximated as follows:
\begin{eqnarray}\label{OM3}
\mu_X(K(z,\epsilon))\propto\mu_{W^c}(K(x_0,\epsilon))\exp\{-\frac{1}{2}\int_s^uOM(\dot{z},z)dt\},
\end{eqnarray}
where symbol $\propto$ denotes the equivalence relation for $\epsilon$ small enough. Here, the $W_t^c$ is defined by
\begin{eqnarray}\label{W}
dW_t^c=cdB_t+\int_{|x|<1}x\tilde{N}(dt,dx).
\end{eqnarray}
\end{theorem}

\begin{proof} Following Definition \ref{Definition 2.1}, we have to consider
\begin{equation}
\mu_X(K(z,\epsilon))=\int_{K(z,\epsilon)}d\mu_X(x).
\end{equation}
For every differentiable function $z(t)\in D_\tau^{x_0}$, we can find a function $z_0(t)\in D_\tau^{0}$ as in (\ref{t1}) such that
\begin{equation}\label{Equation-xm}
z(t)=x_0+z_0(t),\dot{z}(t)=\dot{z_0}(t).
\end{equation}
Thus, the translation $\mathcal{T}$ given by (\ref{t1}) is defined by this $z_0(t)$, and we have
\begin{equation}
\mathcal{T}^{-1}(K(z,\epsilon))=K(x_0,\epsilon).
\end{equation}\par
Note that each jump-diffusion process $Y_t$ with diffusion $c$ and initial value $Y_s=x_0$ induces a quasi-translation invariant measure $\mu_Y$ with $\mu_X\sim\mu_Y$, according to Lemma \ref{lemma 3.1} and Lemma \ref{lemma 3.2}. Hence combining (\ref{F}) and (\ref{y}), we get
\begin{align}\label{hj}
\mu_X(K(z,\epsilon))=\int_{K(z,\epsilon)}F[y]d\mu_Y(y)=\int_{K(x_0,\epsilon)}F[x+z_0]J_Y[x,-z_0]d\mu_Y(x)
\end{align}
as $K(z,\epsilon)\in \mathcal{M} \subset \mathcal{P}.$
In particular, as $\mu_X$ is already quasi-translation invariant by Lemma \ref{lemma 3.2}, we may take $\mu_X$ instead of $\mu_Y$ in (\ref{hj}). Then $F=1$ and instead of (\ref{hj}), we obtain
\begin{eqnarray}\label{hjj}
\mu_X(K(z,\epsilon))=\int_{K(x_0,\epsilon)}J_X[x,-z_0]d\mu_X(x).
\end{eqnarray}

To shift the integration domain to $K(0,\epsilon)$, we further define
\begin{eqnarray}
Y_t^0=Y_t-x_0
\end{eqnarray}
and denote by $\mu_{Y^0}$ the measure induced on $D_\tau^{0}$.\par

\noindent Applying relation  (\ref{y}) again with translation parameter $x_0$ yields
\begin{eqnarray}\label{hjf}
\mu_X(K(z,\epsilon))=\int_{K(0,\epsilon)}F[y+z]J_Y[y+x_0,-z_0]d\mu_{Y^0}(y).
\end{eqnarray}
By combining (\ref{F}) and (\ref{F2}) we can write for the integrand of (\ref{hjf})
\begin{align}\label{FJ}
&F[y+z]J_Y[y+x_0,-z_0]\notag\\
=&\exp\{V(y(u)+z(u))-V(x_0)-\frac{1}{2}\int_s^ub(y(t-)+z(t-))dt\notag\\
&-\Sigma_{s\leq t \leq u}[V(y(t)+z(t))-V(y(t-)+z(t-))]\chi_{|\xi|<1}(\Delta \{y(t)+z(t)\})\notag\\
&+\int_s^u\int_{|\xi|<1}\frac{\xi}{c}a(y(t-)+z(t-))\nu(d\xi)dt\}\notag\\
&\cdot\exp\{V_Y(y(u)+x_0,-z_0(u))-V_Y(x_0,-z_0(s))-\frac{1}{2}\int_s^ud_Y(y(t-)+x_0,-z_0(t))dt\notag\\
&-\Sigma_{s\leq t \leq u}[V_Y(y(t)+x_0,-z_0(t))-V_Y(y(t-)+x_0,-z_0(t))]\chi_{|\xi|<1}(\Delta y(t))\notag\\
&+\int_s^u\int_{|\xi|<1}\frac{\xi}{c}a_Y(y(t-)+x_0,-z_0(t))\nu(d\xi)dt\}.
\end{align}
It is worthwhile to mention that the integrals with respect to time variable $t$ in (\ref{FJ}) are Riemann integrals. Using Taylor expansions, these integrals can be estimated. Now we expand the exponent of (\ref{FJ}) into a Taylor series around $y(t)=0$ and $y(t-)=0$ respectively, and split off the terms of zero order.  If we choose $\epsilon$ small enough, the remaining terms can be made arbitrarily small, as for $y(t)\in K(0,\epsilon), y(t-)\in K(0,\epsilon)$
$$\|y(t)\|\leq\epsilon,\;\;\|y(t-)\|\leq\epsilon$$
holds.\par

Note that $z(t)$ is continuously differentiable. Denoting the remaining terms by $\Delta[y,z]$, we have
\begin{align}\label{FJ1}
&F[y+z]J_Y[y+x_0,-z_0]\notag\\
=&\exp(\Delta[y,z])
\cdot \exp\{V(z(u))-V(x_0)-\frac{1}{2}\int_s^ub(z(t))dt\notag\\&+\int_s^u\int_{|\xi|<1}\frac{\xi}{c}a(z(t))\nu(d\xi)dt\}\notag\\&
\exp\{V_Y(x_0,-z_0(u))-V_Y(x_0,-z_0(s))
-\frac{1}{2}\int_s^ud_Y(x_0,-z_0(t))dt
\notag\\&+\int_s^u\int_{|\xi|<1}\frac{\xi}{c}a_Y(x_0,-z_0(t))\nu(d\xi)dt\}.
\end{align}
By replacing this  (\ref{FJ1}) into (\ref{hjf}), we obtain
\begin{eqnarray}\label{approx}
\mu_X(K(z,\epsilon))=F[z]J_Y[x_0,-z_0]\int_{K(0,\epsilon)}\exp(\Delta[y,z])d\mu_{Y^0}(y).
\end{eqnarray}
Now we deal with the remaining terms by $\Delta[y,z]$ appeared in (\ref{approx}). Recall the fact that if for a functional $\Psi[y]$ on $D_\tau^{x_0}$ with $\|\Psi[y]\|\leq\gamma$, then the following relation holds
$$\int_{B}\Psi[y]d\mu(y)\leq\gamma\mu(B),\;\;B\in\mathcal{P}.$$
So we can choose an $\epsilon>0$ such that $\Delta[y,z]<\gamma$ for $\gamma\rightarrow0$. Then, by expanding the exponential in
(\ref{approx}) and neglecting terms smaller than $\gamma(\gamma\ll1)$, we can approximate (\ref{approx})
\begin{eqnarray}\label{appro}
\mu_X(K(z,\epsilon))\propto F[z]J_Y[x_0,-z_0]\mu_{Y^0}(K(0,\epsilon)).
\end{eqnarray}
Using $\mu_{Y^0}(K(0,\epsilon))=\mu_{Y}(K(x_0,\epsilon))$, we finally get
\begin{eqnarray}\label{tubePf}
\mu_X(K(z,\epsilon))\propto F[z]J_Y[x_0,-z_0]\mu_{Y}(K(x_0,\epsilon)).
\end{eqnarray}
Here, the symbol $\propto$ can be understood as being equal in the sense of $\epsilon$ small enough.
Based on the expression (\ref{tubePf}), in order to find a $z(t)$ which maximizes (\ref{tubePf}), we have to maximize the functional
\begin{eqnarray}\label{Mz}
M[z]&=&F[z]J_Y[x_0,-z_0]\notag\\
&=&\exp\{V(z(u))-V(x_0)-\frac{1}{2}\int_s^ub(z(t))dt\notag\\&&+\int_s^u\int_{|\xi|<1}\frac{\xi}{c}a(z(t))\nu(d\xi)dt\}\notag\\&&
\exp\{V_Y(x_0,-z_0(u))-V_Y(x_0,-z_0(s))
-\frac{1}{2}\int_s^ud_Y(x_0,-z_0(t))dt
\notag\\&&+\int_s^u\int_{|\xi|<1}\frac{\xi}{c}a_Y(x_0,-z_0(t))\nu(d\xi)dt\}.
\end{eqnarray}

Now let us simplify the function (\ref{Mz}) in the same spirit of diffusion process situation by \cite{DB}. $X_t$ is given by (\ref{Equation-xmm}) and let $Y_t$ be given by (\ref{Equation-y}) with $g(x)=c, h(y,x)=x$. To get (\ref{Mz}) in terms of $f(x)$ and $k(x)$, we need the following equations for which we used (\ref{a}), (\ref{Va}), (\ref{b}), (\ref{aX}), (\ref{dX}),(\ref{Equation-xm}) and the differentiability of $z(t)$
\begin{eqnarray}\label{Vf}
V(z(u))-V(x_0)=\frac{1}{c^2}\int_s^u[\dot{z}(t)\{f(z(t))-k(z(t))\}]dt,
\end{eqnarray}
\begin{eqnarray}\label{bf}
b(z(t))&=&\{\frac{f(z(t))-k(z(t))}{c}\}^2+\frac{df(x)}{dx}\mid_{x=z(t)}\notag\\&&-\frac{dk(x)}{dx}\mid_{x=z(t)}
+\frac{2}{c^2}\{f(z(t))-k(z(t))\}k(z(t)),
\end{eqnarray}
\begin{eqnarray}\label{VYf}
V_Y(x_0,-z_0(u))-V_Y(x_0,-z_0(s))=\frac{1}{c^2}\int_s^u[\dot{z}(t)k(z(t))-\ddot{z}(t)x_0\}]dt,
\end{eqnarray}
\begin{eqnarray}\label{dYf}
d_Y(x_0,-z_0(t))&=&-\frac{k^2(x_0)}{c^2}-\frac{d k(x)}{dx}\mid_{x=x_0}+\frac{{\dot{z}(t)}^2}{c^2}\notag\\&&+\frac{k^2(z(t))}{c^2}-\frac{2\ddot{z}(t)x_0}{c^2}+\frac{d k(x)}{dx}\mid_{x=z(t)}.
\end{eqnarray}
We get for $M[z]$ by combining the last four formulas
\begin{align}\label{lnM}
\ln(M[z])=&-\frac{1}{2}\int_s^u\{(\frac{f(z)-\dot{z}}{c})^2+f'(z)-(\frac{k^2(x_0)}{c^2}+k'(x_0))\}dt\notag\\&+
\int_s^u\int_{|\xi|<1}\frac{\xi}{c^2}f(z)\nu(d\xi)dt-\int_s^u
\int_{|\xi|<1}\frac{\xi}{c^2}(k(x_0)+\dot{z})\nu(d\xi)dt.
\end{align}
In accordance with Definition \ref{Definition 2.1}, we define the following OM function:
\begin{equation}\label{OMOO}
OM(\dot{z},z)=(\frac{\dot{z}-f(z)+\int_{|\xi|<1}\xi\nu(d\xi)}{c})^2+f'(z)-\{(\frac{-k(x_0)+\int_{|\xi|<1}\xi\nu(d\xi)}{c})^2+k'(x_0)\}.
\end{equation}
\noindent Notice that the function $z_m(t)$ that maximizes (\ref{tubePf}) must be independent of the choice of the quasi-translation invariant measure $\mu_Y$ i.e. independent of $k(z)$. This is fulfilled for our setting. And the term
$$(\frac{\int_{|\xi|<1}\xi\nu(d\xi)}{c})^2-(\frac{-k(x_0)+\int_{|\xi|<1}\xi\nu(d\xi)}{c})^2+k'(x_0)$$
is a constant and depends only on the measure chosen.
So up to an additive constant, we can take as OM function the expression
\begin{align}\label{OMP}
OM(\dot{z},z)=(\frac{\dot{z}-f(z)}{c})^2+f'(z)+2\frac{\dot{z}-f(z)}{c^2}\int_{|\xi|<1}\xi \nu(d\xi).
\end{align}
We can also rewrite the OM function given in (\ref{OMP}) as follows:
\begin{align}\label{OMR}
OM(\dot{z},z)=\frac{1}{c^2}[\dot{z}-f(z)+\int_{|\xi|<1}\xi\nu(d\xi)]^2+f'(z)-(\frac{\int_{|\xi|<1}\xi\nu(d\xi)}{c})^2.
\end{align}
Now with the help of (\ref{OMP}), we can give some versions of (\ref{tubePf}):
\begin{align}\label{OM1}
\mu_X(K(z,\epsilon))\propto& \exp\{\frac{1}{2}(\frac{k^2(x_0)}{c^2}+k'(x_0)-\frac{2k(x_0)}{c^2}\int_{|\xi|<1}\xi\nu(d\xi))(u-s)\}\notag\\&\cdot\mu_Y(K(x_0,\epsilon))\exp\{-\frac{1}{2}\int_s^uOM(\dot{z},z)dt\},
\end{align}
\begin{align}\label{OM2}
\mu_X(K(z,\epsilon))\propto&\exp\{\frac{1}{2}(\frac{f^2(x_0)}{c^2}+f'(x_0)-\frac{2f(x_0)}{c^2}\int_{|\xi|<1}\xi\nu(d\xi))(u-s)\}\notag\\&\cdot\mu_X(K(x_0,\epsilon))\exp\{-\frac{1}{2}\int_s^uOM(\dot{z},z)dt\},
\end{align}
\begin{align}\label{OM3}
\mu_X(K(z,\epsilon))\propto\mu_{W^c}(K(x_0,\epsilon))\exp\{-\frac{1}{2}\int_s^uOM(\dot{z},z)dt\}.
\end{align}
In the last expression the $W_t^c$ is defined by
\begin{eqnarray}\label{W}
dW_t^c=c\;dB_t+\int_{|x|<1}x\tilde{N}(dt,dx).
\end{eqnarray}
The proof is complete.
\end{proof}
\begin{remark}\label{Remark 4.1} Compared to Brownian noise situation, additional term $2\frac{\dot{z}-f(z)}{c^2}\int_{|\xi|<1}\xi\nu(d\xi)$ appears in OM function (\ref{OM}). If the system (\ref{Equation-xmm}) is only driven by Brownian noise, i.e., $\nu=0$,  OM function (\ref{OM}) is consistent with the result of diffusion process by \cite{DB, IS}. \par In addition, we can find: (i) if the L\'{e}vy measure $\nu$ is asymmetric, the nontrivial effects of the additional term on system (\ref{Equation-xmm}) can be shown. Further it makes the most probable tube of system (\ref{Equation-xmm}) different from the case of Brownian case; (ii) if the  L\'{e}vy measure $\nu$ is symmetric, then the integral $\int_{|\xi|<1}\xi\nu(d\xi)=0$ in the sense of Cauchy principal values, i.e., the additional term $2\frac{\dot{z}-f(z)}{c^2}\int_{|\xi|<1}\xi\nu(d\xi)$ vanishes. It shows that the most probable path is the same as the case of Brownian noise.
\end{remark}
\begin{remark}\label{Remark 4.2} Theorem \ref{theorem 4.1} presents the OM function for a scalar stochastic differential equation with a general L\'{e}vy process with jump measure $\nu$ satisfying integral $\int_{|\xi|<1}\xi\nu(d\xi)<\infty$. In particular, the result of Theorem \ref{theorem 4.1} is valid for $\alpha$-stable L\'evy motion with $0<\alpha<1$ as $\int_{|\xi|<1}|\xi|\nu_{\alpha, \beta}(d\xi)<\infty$. Thus, $\int_{|\xi|<1}\xi\nu_{\alpha, \beta}(d\xi)<\infty$.
\end{remark}

\begin{remark}\label{Remark 4.3} The conclusion of the Theorem \ref{theorem 4.1} can be generalized into higher dimensional cases as long as $\int_{|\xi|<1}\xi\nu(d\xi)<\infty$. The Onsager-Machlup function in n-dimensional case is given, up to an additive constant, by $OM(\dot{z},z)=|\frac{f(z)-\dot{z}}{c}|^2+{\rm div}(f)(z)+2\frac{\dot{z}-f(z)}{c^2}\cdot\int_{|\xi|<1}\xi\nu(d\xi)$, where $|\cdot|$ denotes the Euclidean norm, ``$\cdot$'' represents the scalar product of vectors and ${\rm div}(f)(z)=\Sigma_{i=1}^n\frac{\partial}{\partial z_i}f^i(z)$.
\end{remark}

\renewcommand{\theequation}{\thesection.\arabic{equation}}
\setcounter{equation}{0}

\section{The most probable path}

In this section, we restrict ourselves on paths with fixed initial and fixed final point, so that we can gain more information of the jump-diffusion process (\ref{Equation-xmm}).  We directly minimize the OM functional $\int_s^u OM(\dot{z},z)dt$ to find the approximate most probable path when it exists. We can ensure that the OM functional does indeed have a minimizer, at least within an appropriate Sobolev space if we impose some conditions on the drift term $f(x)$ and L\'{e}vy jump measure $\nu$ in (\ref{Equation-xmm}). Denote OM functional by $I[z]:= \int_s^uOM(\dot{z},z)dt$ for convenience. \par

\begin{theorem}\label{theorem 5.1} Denote $\mathcal{A}= \{z\in H^1([s, u]; \mathbb{R}): z(s)=x_0, z(u)=x_1\}$ the admissible set consisting of transition paths. Assume that (i) the drift term $f(x)$ in (\ref{Equation-xmm}) belongs to $C^1(\mathbb{R})$ and is global Lipschitz and its derivative is bounded below, i.e. there exists a constant $M$ such that $f'(x)\geq M$; (ii) jump measure $\nu$ satisfies $(\int_{|\xi|<1}\xi\nu(d\xi))^2\geq Mc^2$. Then there exists at least one function $z^\ast \in \mathcal{A}$ such that $$I[z^\ast]=\min_{z\in \mathcal{A}}I[z].$$ The $z^\ast$ defines the most probable path connecting $x_0$ and $x_1$.
\end{theorem}
\begin{proof}
The OM function given in (\ref{OM}) can be rewritten as follows:
\begin{align}\label{OMR}
OM(\dot{z},z)=\frac{1}{c^2}[\dot{z}-f(z)+\int_{|\xi|<1}\xi\nu(d\xi)]^2+f'(z)-(\frac{\int_{|\xi|<1}\xi\nu(d\xi)}{c})^2.
\end{align}
By the proof of Lemma 4.2 of \cite{Wan}, we obtain:
$$\int_s^u\frac{1}{c^2}[\dot{z}-f(z)+\int_{|\xi|<1}\xi\nu(d\xi)]^2dt\geq \frac{C_2^{-1}}{c^2}\int_s^u\dot{z}^2dt-\frac{C_1C_2^{-1}}{c^2}(f(0)-\int_{|\xi|<1}\xi\nu(d\xi))^2,$$ where $C_1$ and $C_2$ are two positive constants depending on Lipschitz constants of $f(x)$ and $\tau=[s,u]$.\par
\noindent Thus, we have:
$$I[z]\geq\frac{C_2^{-1}}{c^2}\int_s^u\dot{z}^2dt-\frac{C_1C_2^{-1}}{c^2}(f(0)-\int_{|\xi|<1}\xi\nu(d\xi))^2-(\frac{(\int_{|\xi|<1}\xi\nu(d\xi))^2}{c^2}-M)(u-s)$$
Hence, the coerciveness on $I[\cdot]$ follows. On the other hand, $OM(\dot{z},z)$ is convex in the variable $\dot{z}$. By Theorem 2 of \cite[ Page 470] {Evans}, we reach the conclusion. The proof is complete.
\end{proof}

In addition, if we restrict ourselves to twice differentiable functions $z(t)$, the most probable path  $z_m(t)$ can be found by variation of a functional $\int_s^u OM(\dot{z},z)dt$, which is given as follows:
\begin{eqnarray}\label{varia}
\delta\int_s^uOM(\dot{z},z)dt=0,
\end{eqnarray}
where
\begin{eqnarray}\label{case1}
z_m(s)=x_0, z_m(u)=x_1,\;\;\; x_1\in \mathbb{R}.
\end{eqnarray}
We further get the Euler-Lagrange equation
\begin{eqnarray}\label{EL}
\frac{d}{dt}\frac{\partial OM(\dot{z},z)}{\partial\dot{z}}=\frac{\partial OM(\dot{z},z)}{\partial z}
\end{eqnarray}
as an ordinary differential equation for $z_m(t)$. With the OM function given in (\ref{OM}), we have:
\begin{equation}\label{evolu1}
\ddot{z}_m=\frac{c^2}{2}f''(z_m)+f'(z_m)f(z_m)-f'(z_m)\int_{|\xi|<1}\xi\nu(d\xi)
\end{equation}
with boundary conditions
\begin{equation}\label{condition1}
z_m(s)=x_0,    \;\;\;   z_m(u)=x_1.
\end{equation}

As we have shown, we could determine a most probable tube $K(z_m,\epsilon)$ by means of a variation principle (\ref{varia}) where $\epsilon$ must be smaller than a given $\gamma$. Then equations (\ref{evolu1})-(\ref{condition1}) hold for each $\epsilon<\gamma$. It is worth noting that the function $z_m(t)$ we have found is not a real orbit of system (\ref{Equation-xmm}). It is a reference path and the probability that the   jump-diffusion process stay in a tube around $z_m(t)$ is maximal or local maximal in the sense of Theorem \ref{theorem 5.2}.\par

\begin{theorem}\label{theorem 5.2} If the solution of Euler-Lagrange equation (\ref{evolu1})-(\ref{condition1}) is smooth and function $OM(\dot{z},z)$ is convex in the variable $\dot{z}$, then this solution is indeed a local minimizer of OM functional $I[z]$. Furthermore, this solution is in fact a global minimizer if the joint mapping $(\dot{z}, z)\longmapsto OM(\dot{z},z)$ is convex.
\end{theorem}
\begin{proof} Note that our OM function given in (\ref{OMR}) is convex in the variable $\dot{z}$. Thus, by Theorem 10 in Section 8.2.5 of \cite[ Page 481] {Evans}, the first conclusion holds. And by the Remark in Section 8.2.3 of \cite[ Page 474] {Evans}, we reach the second conclusion. The proof is complete.
\end{proof}

\begin{remark}\label{Remark 5.1} Our method to derive the OM functional is restricted on differentiable functions $z(t)$. And once the expression of OM functional is available, it is reasonable to define OM functional $I[\cdot]$ not only for differentiable functions, but also for functions in the Sobolev space. Thus, Theorem \ref{theorem 5.1} offers a sufficient condition to establish a minimizer among the functions in $\mathcal{A}$ and the minimizer of $I[\cdot]$ could give the most probable path for system (\ref{Equation-xmm}). Although the existence theory of minimizers for the OM functional in the space of smooth functions through direct minimization is incomplete, Theorem \ref{theorem 5.2} offers a sufficient condition to obtain the existence of minimizer in $C^2(\mathbb{R})$ through Euler-Lagrange equation.
\end{remark}

The problem of solving (\ref{evolu1})-(\ref{condition1}) is referred to as the two-point boundary value problem. As we know, this problem does not always have a solution. There is the same issue for the case of diffusion processes. The literature on the existences and uniqueness of solutions for (\ref{evolu1})-(\ref{condition1}) can be referred to \cite{HP, Se}.\par

\section{Numerical experiments}

In this section, we choose asymmetric $\alpha$-stable L\'evy motion with $0<\alpha<1$ as the driven L\'evy noise. We illustrate our results analytically and numerically in some concrete examples whose most probable paths can be found by means of Euler-Lagrange equations. In the following, we denote $\int_{|\xi|<1}\xi\nu_{\alpha,\beta}(d\xi)=\frac{\alpha}{\Gamma(2-\alpha)\cos(\frac{\pi\alpha}{2})}\beta$ by $d_{\nu_{\alpha,\beta}}$ and take $s=0$, $u=T$ for convenience. Notations $d, d_j,j=0,1,2,3,4$ appeared in figures indicate the different values of $d_{\nu_{\alpha,\beta}}$, which may change from one place to another.\par
\noindent {\bf Example 1.} (A stochastic system with linear potential)\par
\noindent Consider the following scalar SDE with linear potential:
\begin{equation} \label{example-1}
  \begin{split}
  &dX_t=-X_{t-}dt+dB_t+\int_{|x|<1}x\tilde{N}(dt,dx),    \\
  &f(x)=-x,\; g(x)=1,\; X_0=x_0\in \mathbb{R}.
  \end{split}
\end{equation}
By using (\ref{OM}), we obtain the OM function of this system :
\begin{equation}\label{example-OM}
OM(\dot{z},z)=(z+\dot{z})^2-1+2(\dot{z}+z)\int_{|\xi|<1}\xi\nu_{\alpha,\beta}(d\xi).
\end{equation}
As $f(x)=-x$ satisfies the conditions in the Theorem \ref{theorem 5.1}, the OM functional $\int_0^T OM(\dot{z},z)dt$ does indeed have a minimizer in $\mathcal{A} = \{z\in H^1([0, T];\mathbb{R}) : z(0)=x_0, z(T)=x_1\}$, where the OM function is given in (\ref{example-OM}). And this minimizer gives the most probable path for system (\ref{example-1}).\par
 On the other hand, via (\ref{evolu1})-(\ref{condition1}), the most probable path $z_m(t)$ satisfies the following two-point boundary value problem:
\begin{equation} \label{example-evolu1}
  \begin{split}
  &\ddot{z}_m=z_m+\int_{|\xi|<1}\xi\nu_{\alpha,\beta}(d\xi),    \\
  &z_m(0)=x_0, z_m(T)=x_1,\;\;\; x_1\in \mathbb{R}.
  \end{split}
\end{equation}
Note that the joint mapping $(\dot{z}, z)\longmapsto OM(\dot{z},z)$ given in (\ref{example-OM}) is convex, thus the solution of (\ref{example-evolu1}) is indeed a minimizer of functional $\int_0^T OM(\dot{z},z)dt$ by the Theorem \ref{theorem 5.2}.
According to the variation of constants \cite{HP}, we obtain the explicit solution for this problem (\ref{example-evolu1}):
\begin{equation}\label{example-solution1}
z_m(t)=(c_1+\frac{1}{2}d_{\nu_{\alpha,\beta}})e^t+(c_2+\frac{1}{2}d_{\nu_{\alpha,\beta}})e^{-t}-d_{\nu_{\alpha,\beta}}
\end{equation}
where $c_1, c_2$ are given by
\begin{equation}\label{constants1}
c_1=\frac{\left|\begin{array}{cc}
x_0+d_{\nu_{\alpha,\beta}}-\frac{1}{2}d_{\nu_{\alpha,\beta}}-\frac{1}{2}d_{\nu_{\alpha,\beta}} & 1\\
x_1+d_{\nu_{\alpha,\beta}}-\frac{1}{2}e^Td_{\nu_{\alpha,\beta}}-\frac{1}{2}e^{-T}d_{\nu_{\alpha,\beta}} & e^{-T}
\end{array}\right|}{\left|\begin{array}{cc}
1 & 1\\
e^{T} & e^{-T}
\end{array}\right|},
\end{equation}
\begin{equation}\label{constants2}
c_2=\frac{\left|\begin{array}{cc}
1 & x_0+d_{\nu_{\alpha,\beta}}-\frac{1}{2}d_{\nu_{\alpha,\beta}}-\frac{1}{2}d_{\nu_{\alpha,\beta}}\\
e^T & x_1+d_{\nu_{\alpha,\beta}}-\frac{1}{2}e^Td_{\nu_{\alpha,\beta}}-\frac{1}{2}e^{-T}d_{\nu_{\alpha,\beta}}
\end{array}\right|}{\left|\begin{array}{cc}
1 & 1\\
e^{T} & e^{-T}
\end{array}\right|}
\end{equation}
respectively. Notice that $\left|\begin{array}{cc}
1 & 1\\
e^{T} & e^{-T}
\end{array}\right|$
is not euqal to zero unless $T=0$. Thus, for every given initial point and finial point, we can always find the most probable tube of the system (\ref{example-1}).\par
In Fig \ref{fig1}(a), we show a most probable path of the system (\ref{example-1}) obtained by the shooting method for the two-point boundary value problem, refer to \cite{Se}. The initial value is chosen as $z(0)=3$ and the finial point is chosen as $z(1)=4$. It is seen that the numerical results agree with the theoretical results very well. In Fig \ref{fig1}(b), the green line that corresponds to $d_{\nu_{\alpha,\beta}}=0$ denotes the most probable path of (\ref{example-1}) only driven by
Gaussian noise. And it's clearly different from those paths with nonzero $d_{\nu_{\alpha,\beta}}$. This displays the effect of non-Gaussian noise on the dynamics.  \\
\begin{figure}
 \centering
 \subfigure[]{
 \includegraphics[width=0.45\textwidth]{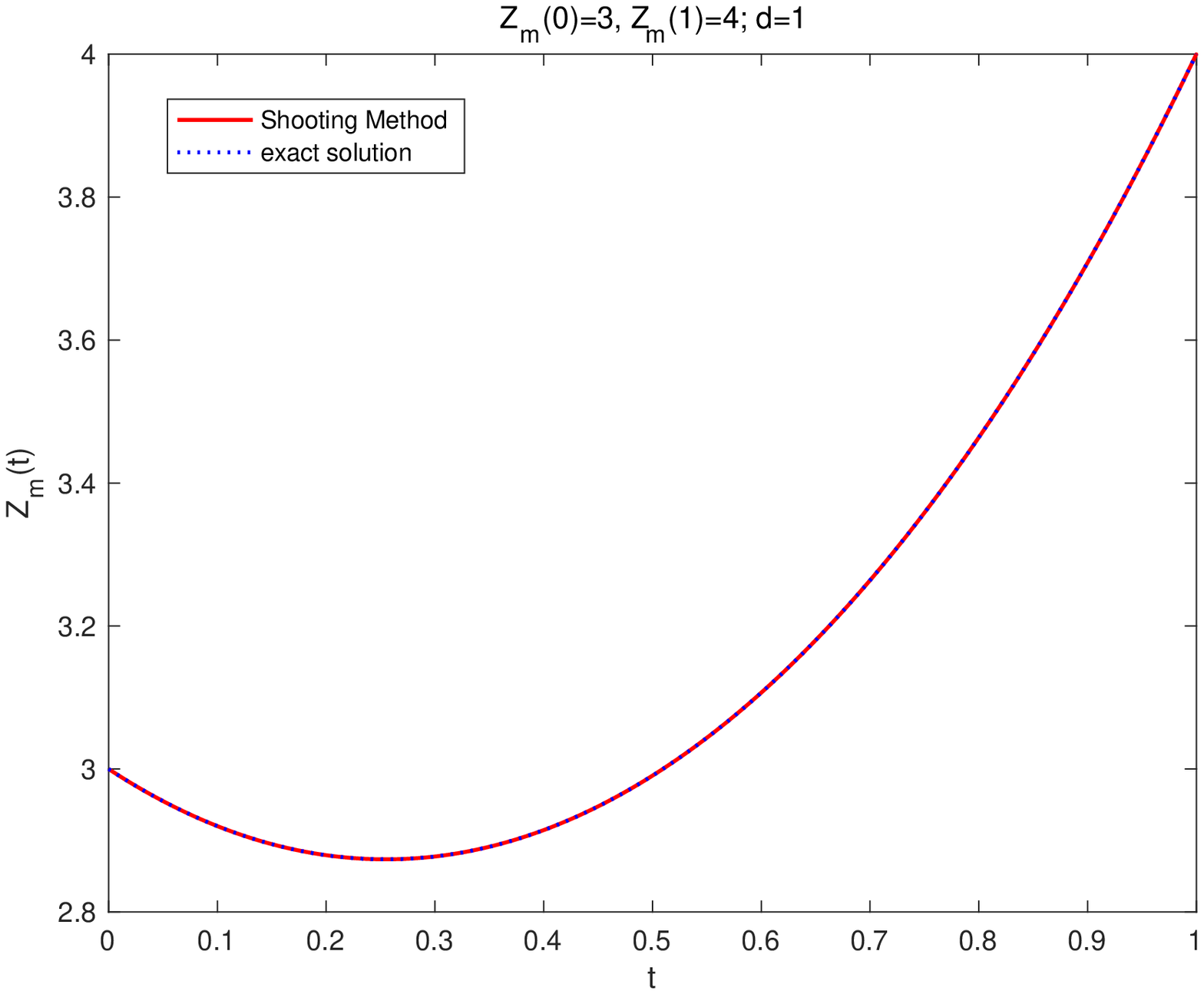}
 }
 \subfigure[]{
 \includegraphics[width=0.45\textwidth]{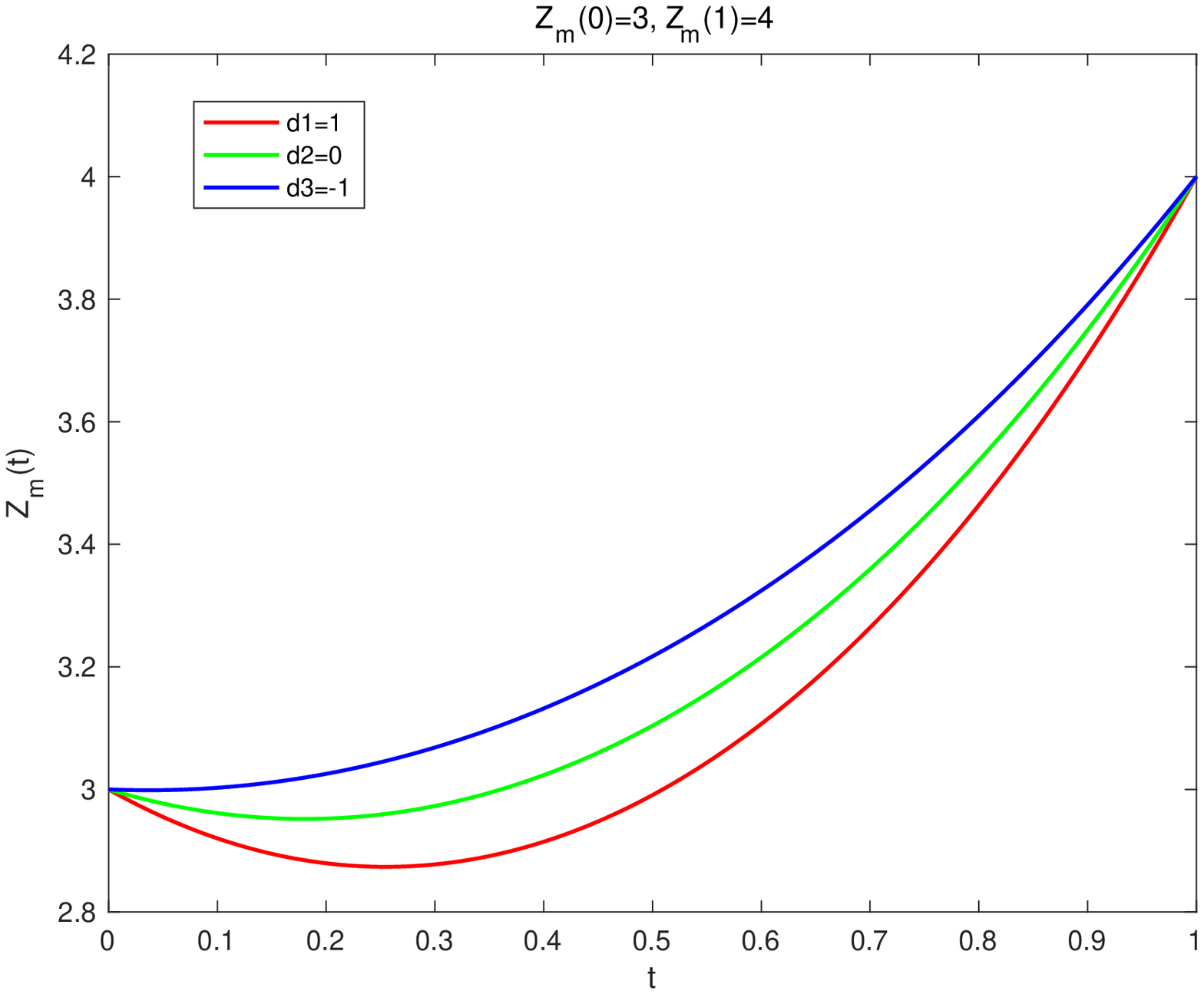}
 }
 \caption{ (Color online) (a) The exact solution (dotted line) versus the numerical solution
(red line) for (\ref{example-evolu1}): Initial value
z(0) = 3, finial value z(1) = 4, and $d_{\nu_{\alpha,\beta}}=1$, (b) Most probable path for system (\ref{example-1}) with respect to time interval parameter and $d_{\nu_{\alpha,\beta}}$: Initial value
z(0) = 3, finial value z(1) = 4.}
 \label{fig1}
 \end{figure}

\noindent {\bf Example 2.} (A stochastic double-well system)\par
\noindent Consider the following scalar SDE
\begin{equation} \label{example-2}
  \begin{split}
  &dX_t=(X_{t-}-X_{t-}^3)dt+dB_t+\int_{|x|<1}x\tilde{N}(dt,dx),    \\
  &f(x)=x-x^3, \;g(x)=1,\; X_0=x_0\in \mathbb{R}.
  \end{split}
\end{equation}
The deterministic counterpart $\dot{x}=x-x^3$ has two stable equilibrium states, -1 and 1, and a unstable equilibrium state, 0. When noise is present, the system (\ref{example-2}) may show a transition from state -1 to state 1 and then these two states are called metastable states. By means of the theory we have developed in the previous section, we can find the most probable path of the stochastic double-well system (\ref{example-2}) among all possible smooth curves connecting these two given points.\par
By Theorem \ref{theorem 4.1}, the OM function of this system is given by
\begin{equation}\label{example-OM-2}
OM(\dot{z},z)=(z-z^3-\dot{z})^2+1-3z^2+2(\dot{z}-z+z^3)\int_{|\xi|<1}\xi\nu_{\alpha,\beta}(d\xi).
\end{equation}
And via (\ref{evolu1})-(\ref{condition1}), the most probable path $z_m(t)$ could be described as the deterministic differential equation:
\begin{equation} \label{example-evolu2}
  \begin{split}
  &\ddot{z}_m=3z_m^5-4z_m^3+3d_{\nu_{\alpha,\beta}}z_m^2-2z_m-d_{\nu_{\alpha,\beta}},    \\
  & z_m(0)=-1, \;\;\;  z_m(T)=1,
  \end{split}
\end{equation}
which can be calculated numerically by the shooting method for the two-point boundary value problem \cite{Se}.
Note that the OM function given in (\ref{example-OM-2}) is convex on the variable $\dot{z}$ but the joint mapping $(\dot{z}, z)\longmapsto OM(\dot{z},z,)$ is not convex, thus the solution of (\ref{example-evolu2}) is a local minimizer of functional $\int_0^T OM(\dot{z},z)dt$ by the Theorem \ref{theorem 5.2}.
We remark that the most probable tube of system (\ref{Equation-xmm}) or, in particular (\ref{example-2}) depends on the choice of time interval, initial state, final state and the driven $\alpha$-stable noise, more precisely, i.e. constant $d_{\nu_{\alpha,\beta}}$. For $\alpha\in(0,1)$ used for our setting, the value of $d_{\nu_{\alpha,\beta}}=\frac{\alpha}{\Gamma(2-\alpha)\cos(\frac{\pi\alpha}{2})}\beta$ can be positive or negative as $\beta\in[-1, 1]$. \par

Fig \ref{fig2} shows the most probable paths with initial state -1 and final state 1 of the system (\ref{example-2}). We choose different values of $T$ to observe the impact of noise on the dynamics. For $T=1$, it is seen that the equation (\ref{example-evolu2}) describing the most probable path can only be solved when $d_{\nu_{\alpha,\beta}}$ is evaluated in the interval $[-1.9407, 1.6305]$. For $T=2$, we obtain an interval $[-1.5150, 0.8487]$ for $d_{\nu_{\alpha,\beta}}$. For $T=2.3$, the most probable path of system (\ref{example-2}) is not available if $d_{\nu_{\alpha,\beta}}$ exceeds the interval $[-1.5057, 1.1396]$. Similarly, for $T=6$, $d_{\nu_{\alpha,\beta}}$ is restricted on interval $[-0.8899, 0.3836]$. Although $T$ could be any finite value from zero to positive infinity, the corresponding value range of $d_{\nu_{\alpha,\beta}}$ does not always exist as equation (\ref{example-evolu2}) is not always solvable. It is seen that for different final time $T$, the most probable paths connecting state -1 and state 1 may take different shapes. For every fixed $T$, the most probable paths are also different from each other depending on  parameter value $d_{\nu_{\alpha,\beta}}$. In particular, both in Fig \ref{fig2}(a), Fig \ref{fig2}(b), Fig \ref{fig2}(c) and Fig \ref{fig2}(d), the green line that corresponds to $d_{\nu_{\alpha,\beta}}=0$ denotes the most probable path of (\ref{example-2}) with absence of L\'{e}vy noise, (i.e. only driven by
Gaussian noise), which is clearly different from those paths with nonzero $d_{\nu_{\alpha,\beta}}$. This finding geometrically displays the effect of non-Gaussian noise on the dynamics. \par


\begin{figure}
 \centering
 \subfigure[]{
 \includegraphics[width=0.45\textwidth]{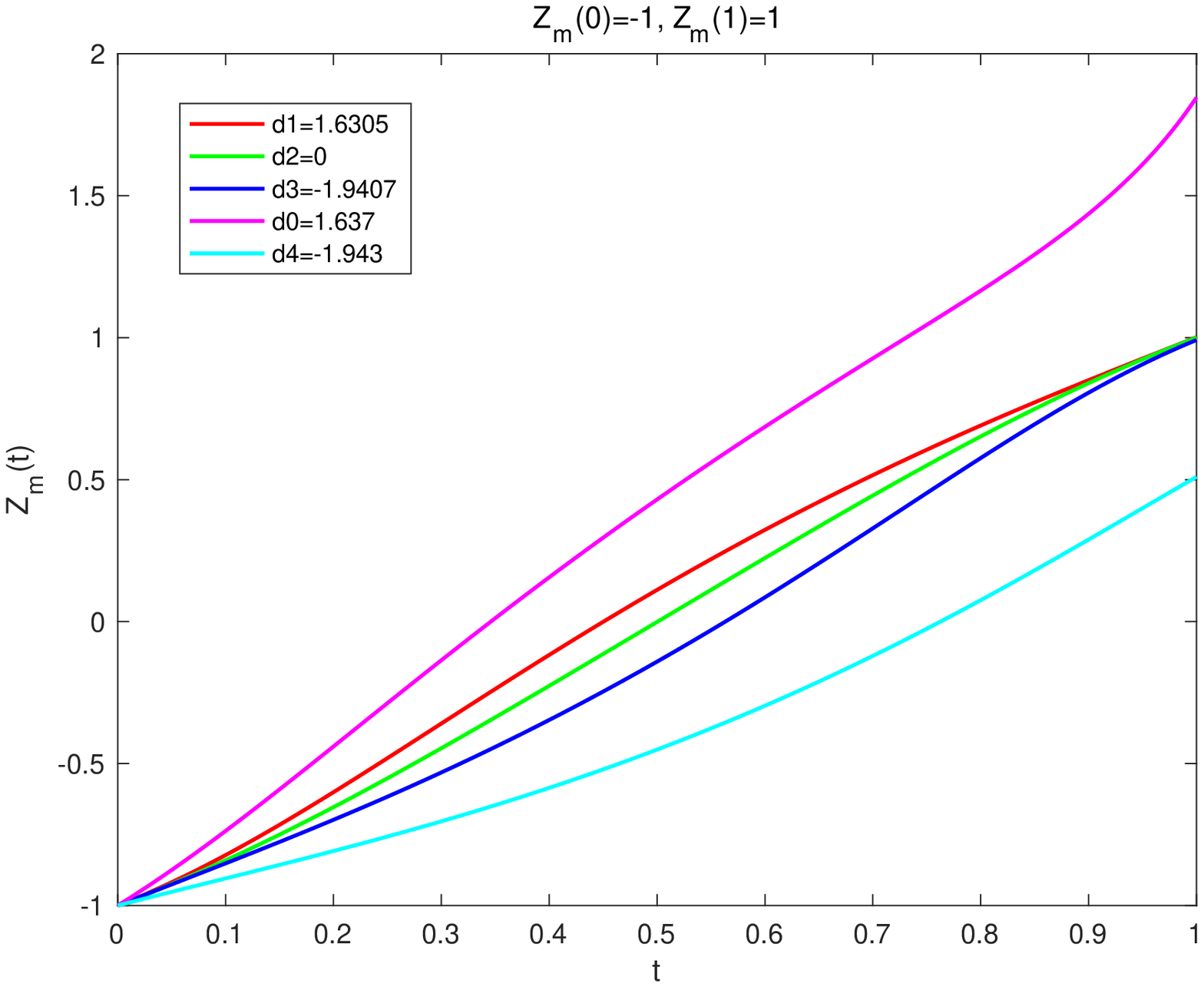}
 }
 \subfigure[]{
 \includegraphics[width=0.45\textwidth]{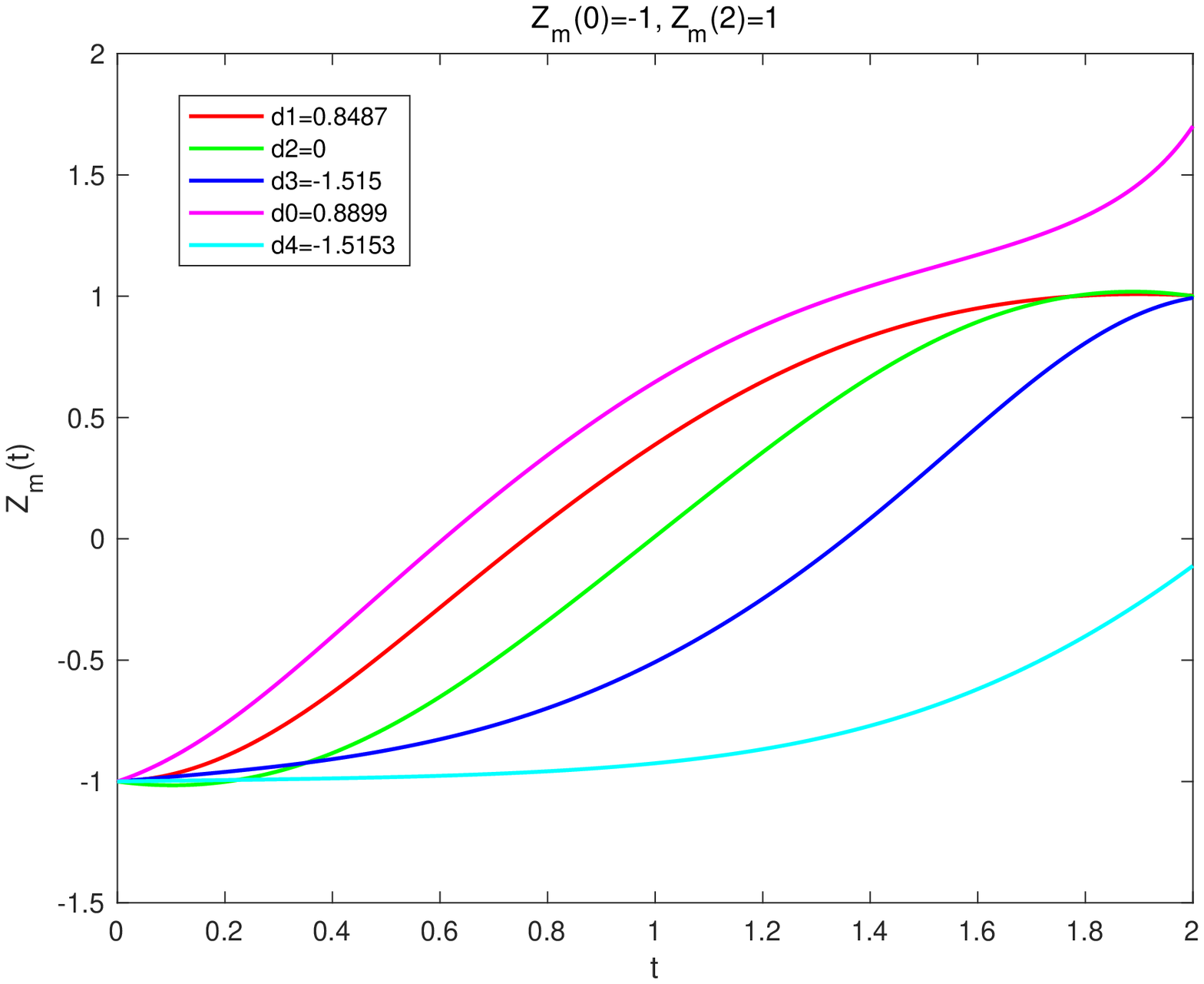}
 }
 \subfigure[]{
 \includegraphics[width=0.45\textwidth]{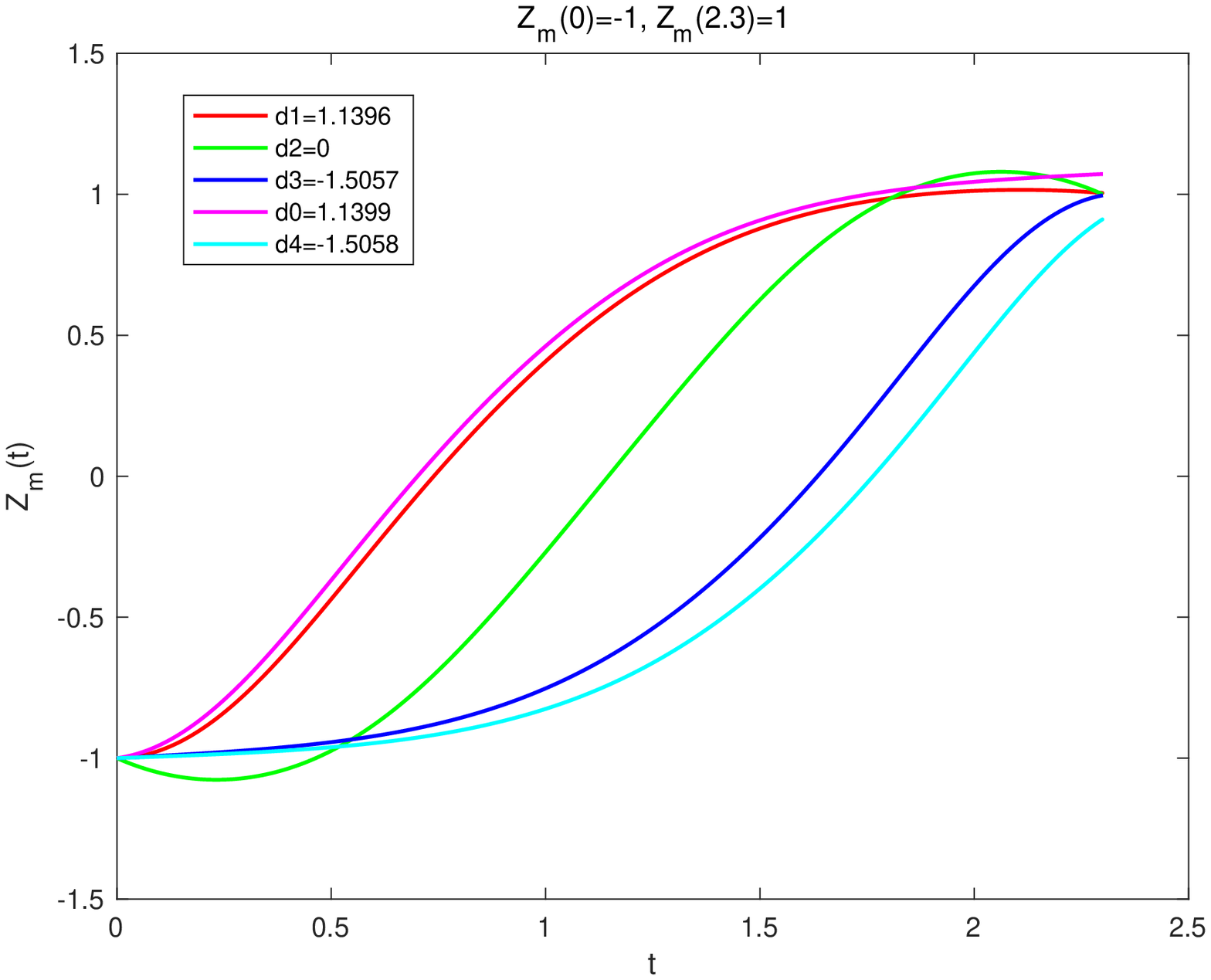}
 }
 \subfigure[]{
 \includegraphics[width=0.45\textwidth]{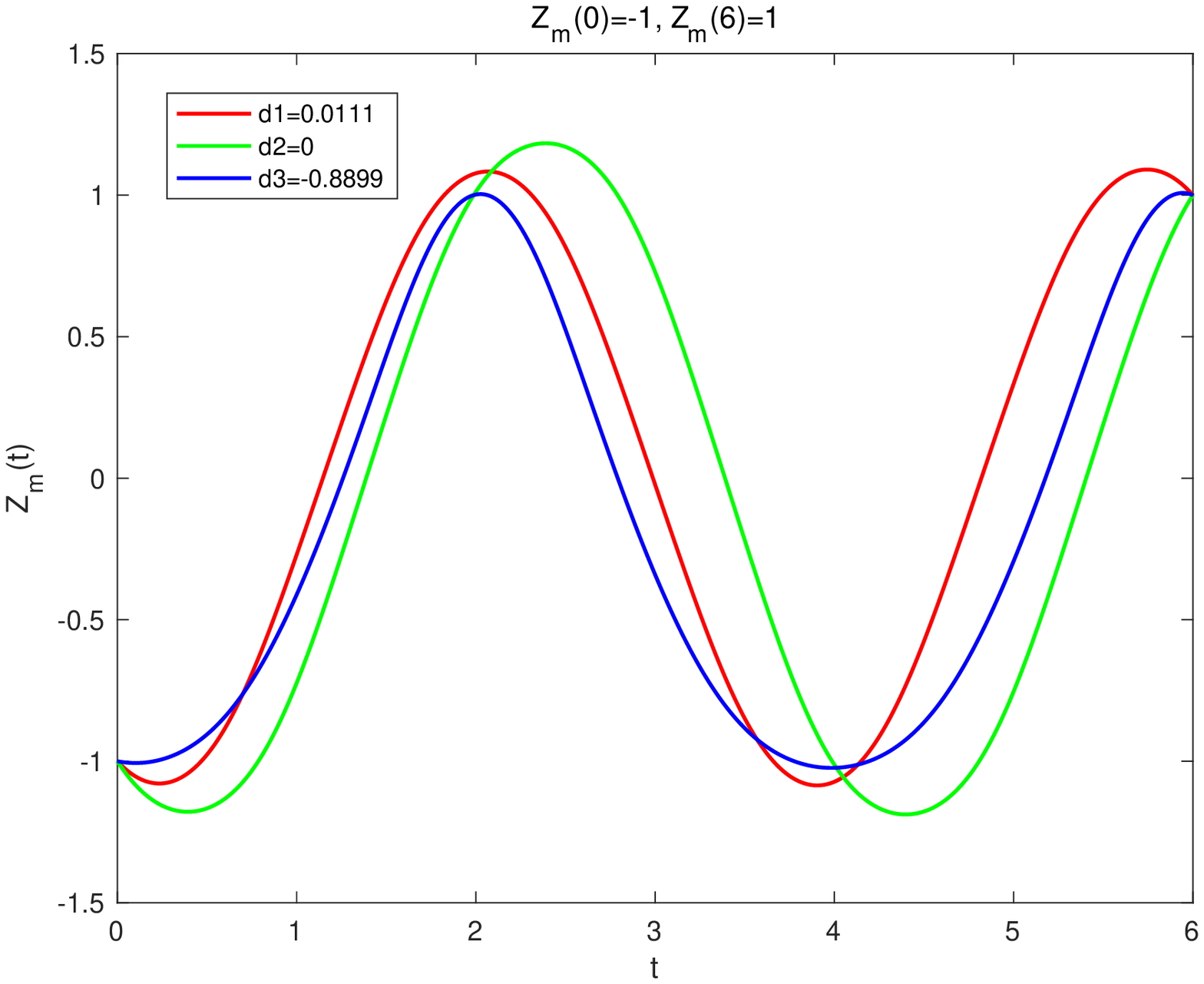}
 }
 \caption{ (Color online) Most probable path for system (\ref{example-2}) with respect to time interval parameter and $d_{\nu_{\alpha,\beta}}$: (a) $T=1$, (b) $T=2$, (c) $T=2.3$, (d) $T=6$.}
 \label{fig2}
 \end{figure}
 Note that choices of $\alpha$ and $\beta$ uniquely determine the value of $d_{\nu_{\alpha,\beta}}$. Thus, Fig \ref{fig3} shows which values of $\alpha$ and $\beta$  correspond to existence (`green') or non-existence (`red') of the most probable paths from state -1 to state 1 of system (\ref{example-2}).
\begin{figure}
 \centering
 \subfigure[]{
 \includegraphics[width=0.45\textwidth]{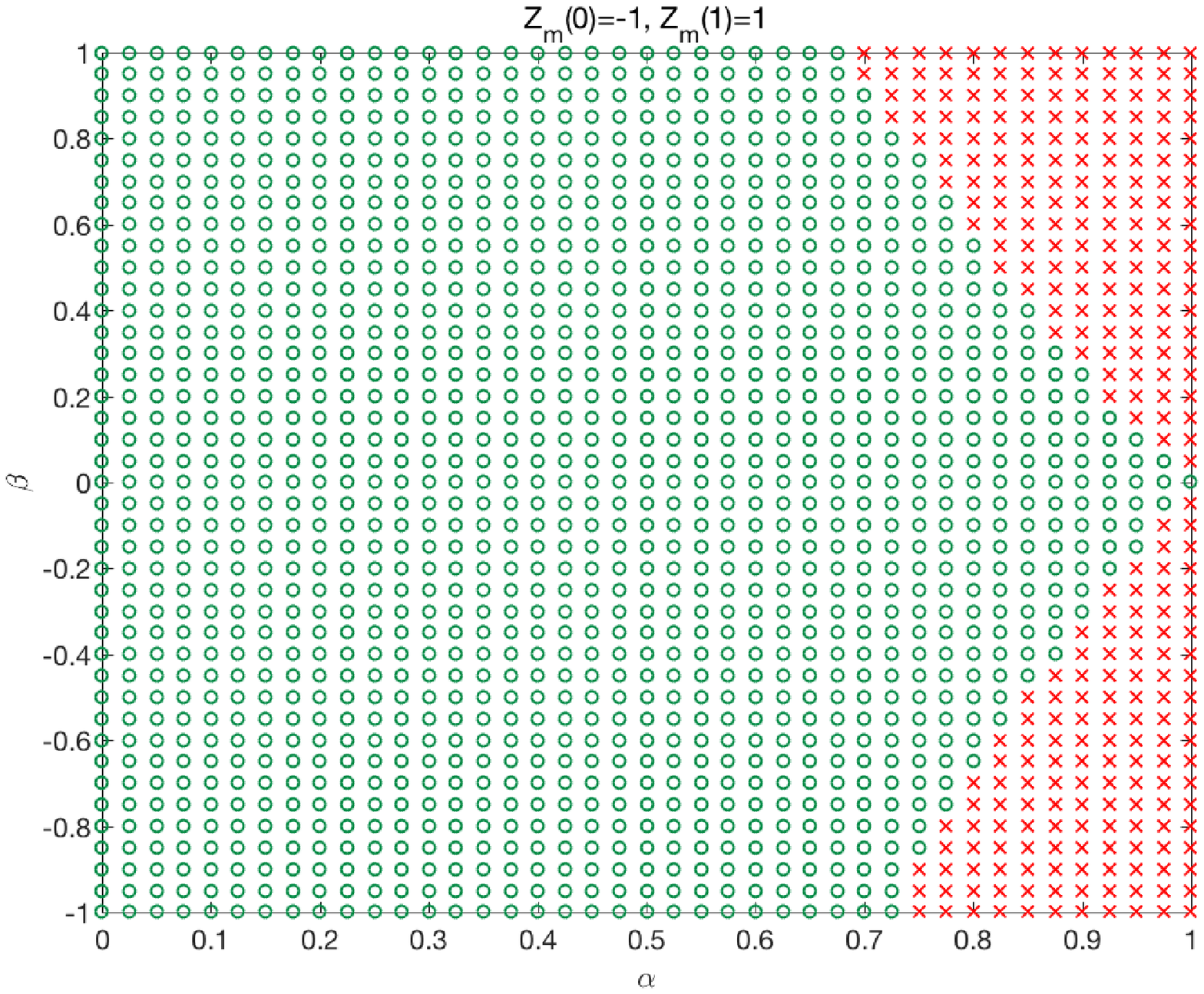}
 }
 \subfigure[]{
 \includegraphics[width=0.45\textwidth]{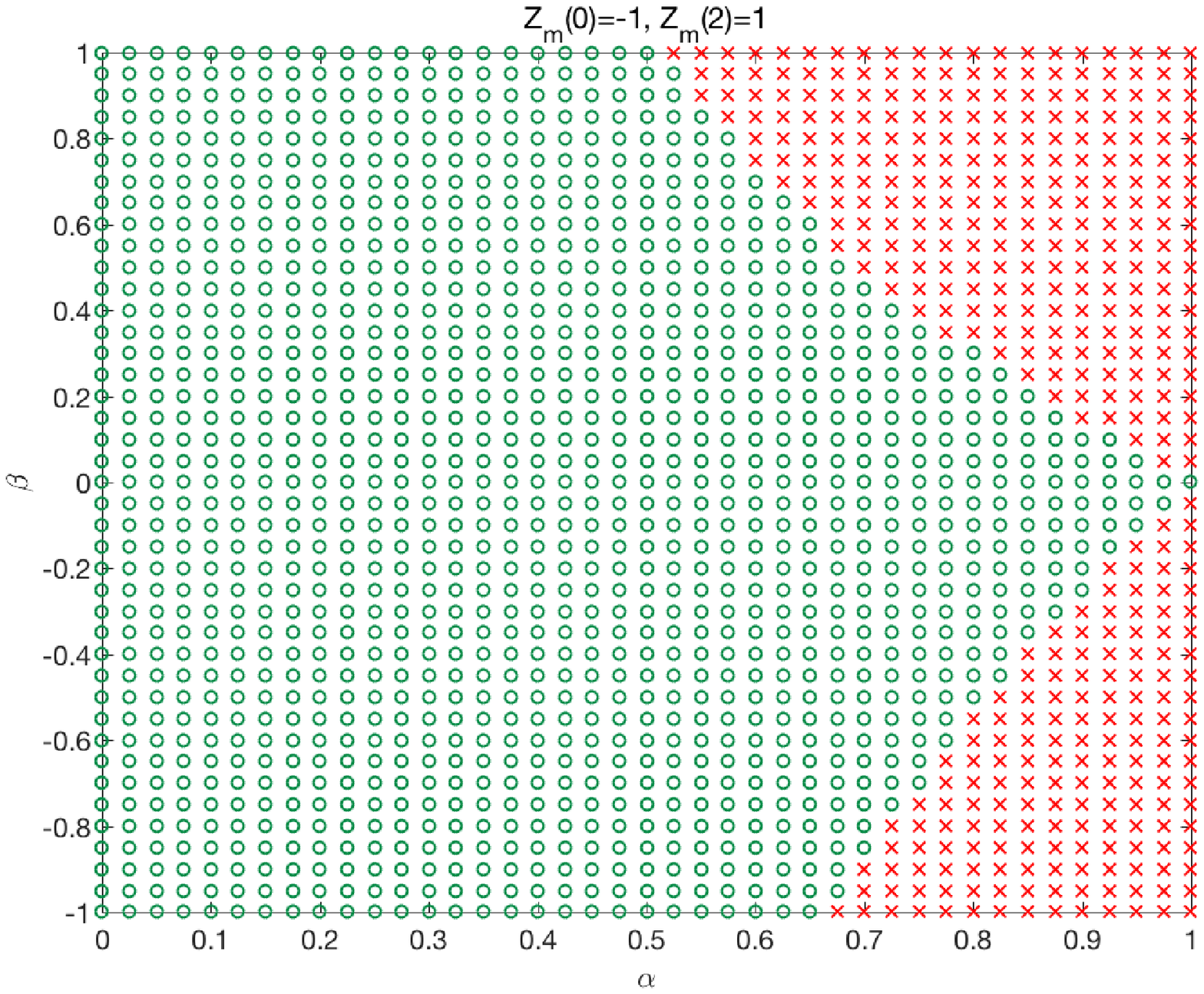}
 }
 \subfigure[]{
 \includegraphics[width=0.45\textwidth]{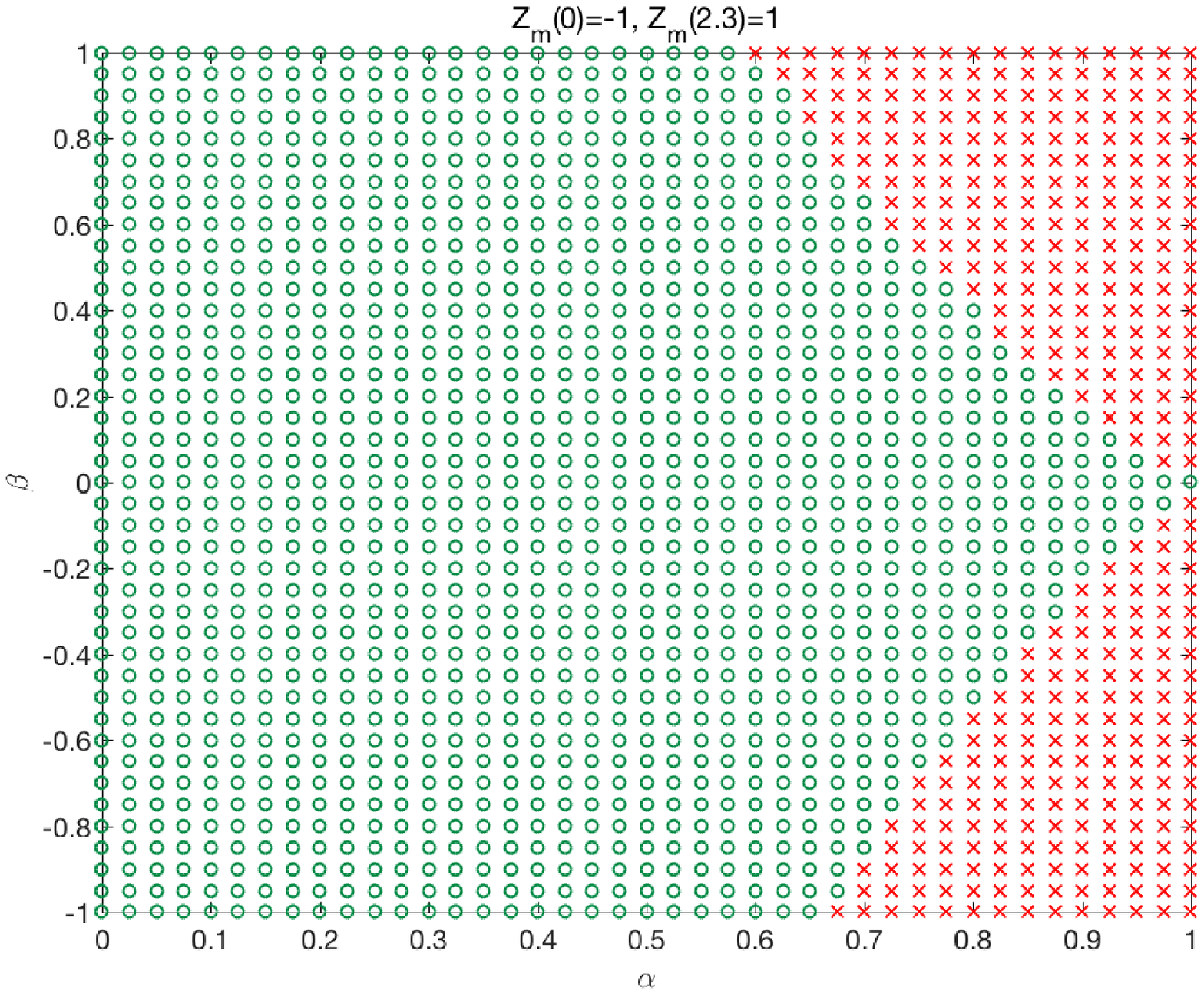}
 }
 \subfigure[]{
 \includegraphics[width=0.45\textwidth]{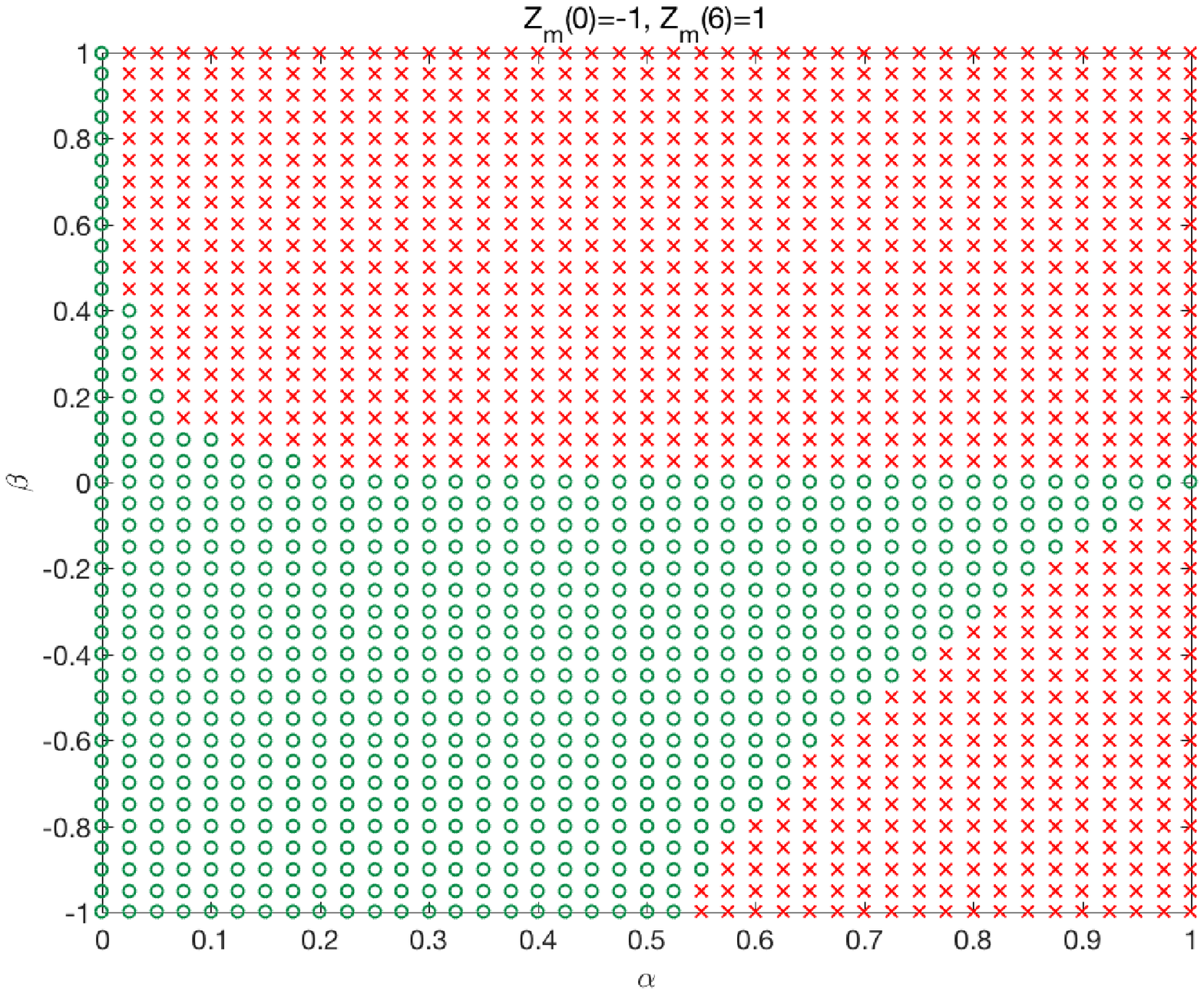}
 }
 \subfigure[]{
 \includegraphics[width=0.45\textwidth]{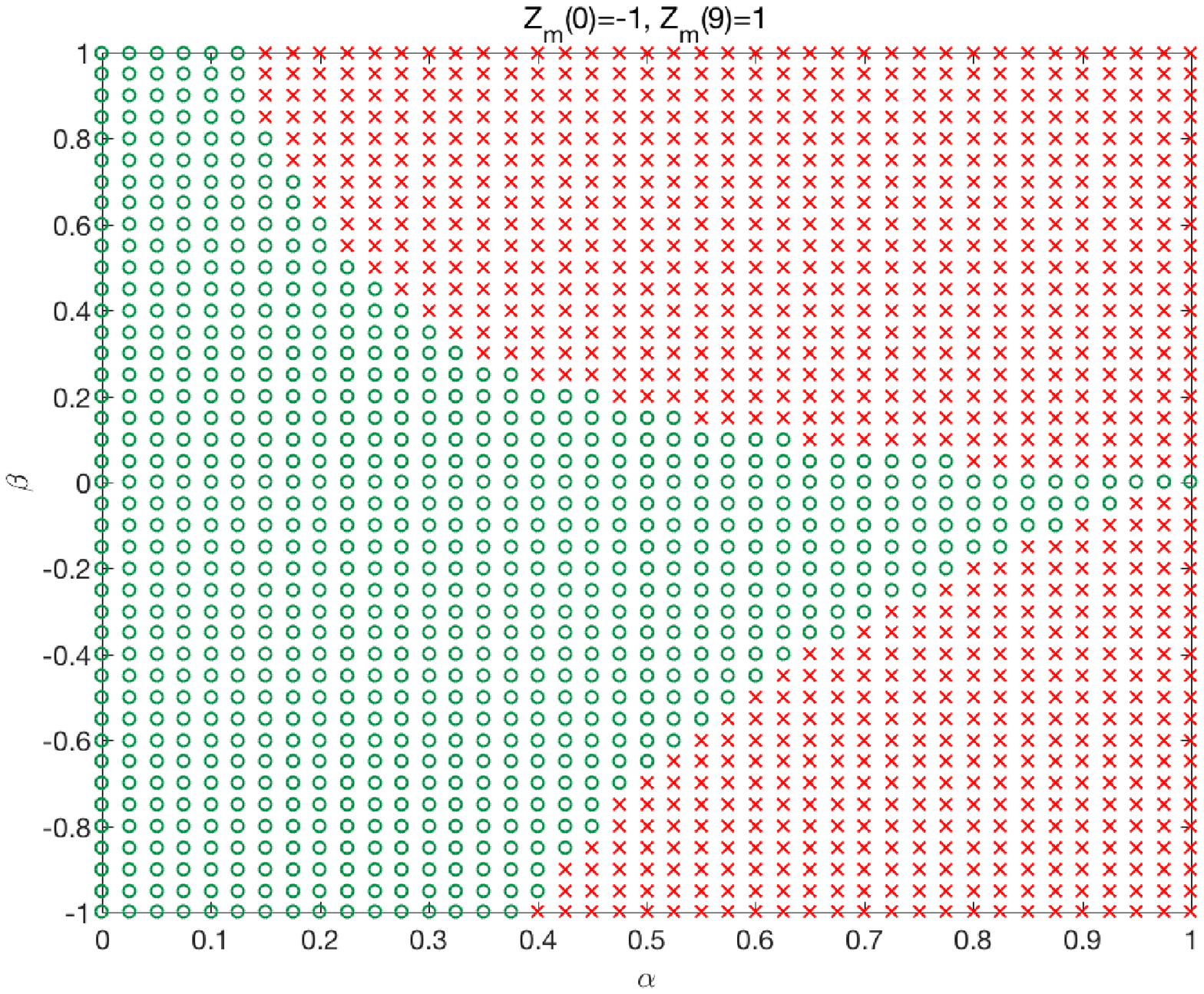}
 }
 \subfigure[]{
 \includegraphics[width=0.45\textwidth]{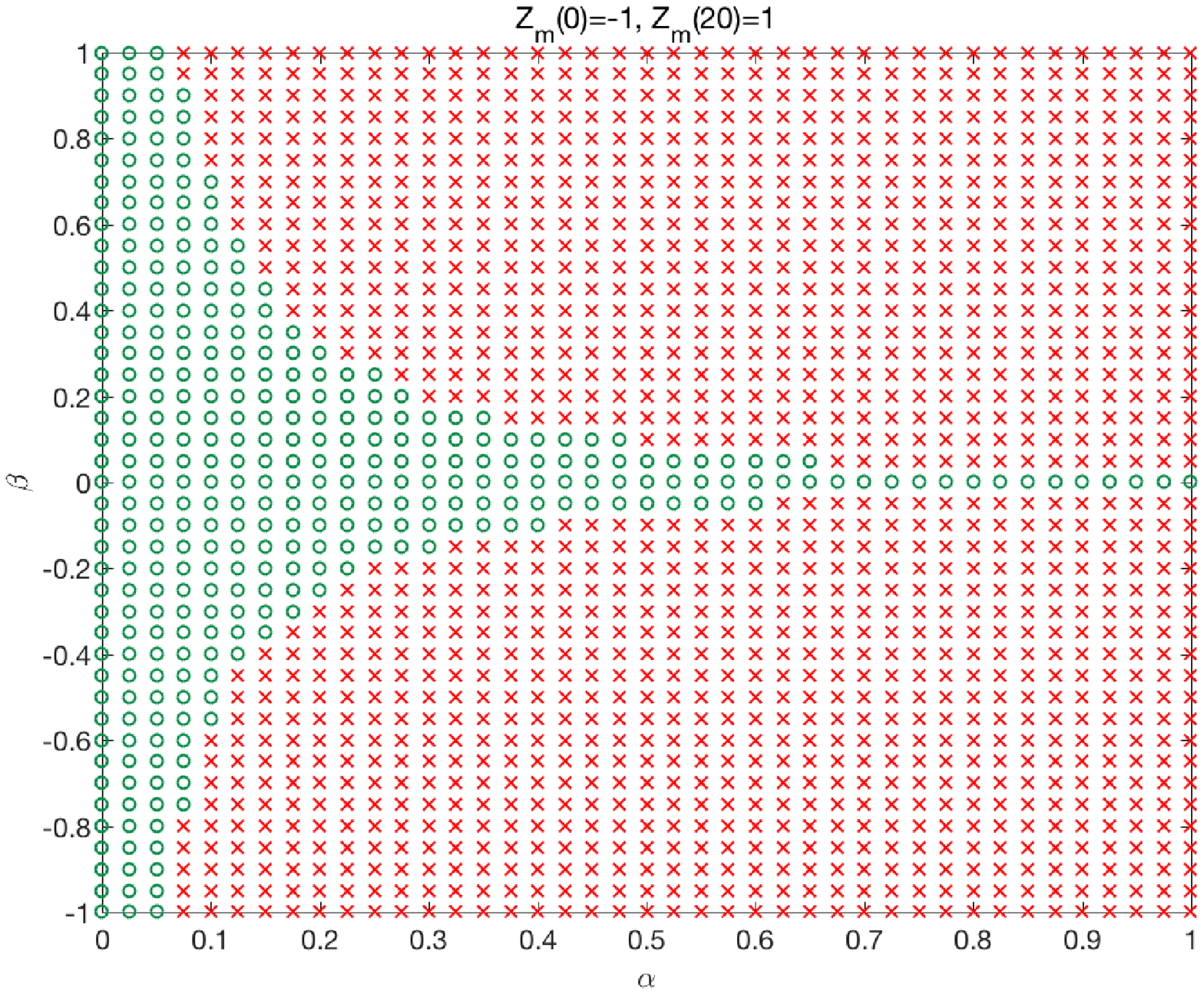}
 }
 \caption{(Color online) The parameter regions of $\alpha$ and $\beta$ with green (red) indicating existence (non-existence) of the most probable paths: (a) $T=1$, (b) $T=2$, (c) $T=2.3$, (d) $T=6$, (e) $T=9$, (f) $T=20$. }
 \label{fig3}
 \end{figure}

\section{Conclusion and discussion }

In this work, we have derived the Onsager-Machlup function for scalar  stochastic differential equations with L\'evy motion as well as Brownian motion. With this function as a Lagrangian, we have characterized the most probable path as the solution of  the corresponding Euler-Lagrange equation, under two-point boundary conditions. This Onsager-Machlup function is consistent with that for stochastic differential equations with  Brownian motion alone, and it thus captures the effect of non-Gaussian fluctuations in understanding the most probable paths.

The most probable paths are often sought,  in order to reveal the most likely pathways in transition phenomena \cite{DB,Duan}. Our work provides a tool to investigate this issue for stochastic dynamical systems with non-Gaussian as well as Gaussian noise, as illustrated in Section 6. Note that this OM theory for the most probable paths does not require noise intensity to be small. Thus it is different from the  most probable paths  based on large deviation principles, which hold for small noise.

The results of this work (Theorem \ref{theorem 4.1}) are valid for a scalar stochastic differential equation with a general L\'evy process with jump measure $\nu$, as long as the integral $\int_{|\xi|<1} \xi \nu(d \xi)$ is finite. In particular, they are valid for the stochastic differential equations with $\alpha$-stable L\'evy motion with $0<\alpha<1$.

Moreover, this work can also be extended to higher dimensional systems (Remark \ref{Remark 4.3}).

However, this approach to derive Onsager-Machlup  action functionals requires non-vanishing and constant diffusion
component in the random  fluctuations, and   appears not to work for     stochastic differential equations with pure jump L\'evy noise. The Girsanov transformation, which plays a crucial role in this approach, is of little help in pure jump   case, since it can not help to obtain the uniqueness in distribution. As also realized earlier in \cite{XCS} that the Girsanov transformation appears not to work for their case, since it transforms the Poisson process to a general semi-martingale that cannot be handled easily.


\renewcommand{\theequation}{\thesection.\arabic{equation}}
\setcounter{equation}{0}
\section*{Appendix: Proof of (\ref{F}) and (\ref{F2}) }
We now give the proof of the technical formulas (\ref{F}) and (\ref{F2}) which are based on the It\^o formula for stochastic integrals.
\begin{proof}[{\bf Proof of (\ref{F})}] By using It\^o formula for stochastic integrals, we get
\begin{align} dV(Y_t)
=&\{\frac{1}{c}a(Y_{t-})k(Y_{t-})+\frac{c}{2}\frac{da(x)}{dx}(Y_{t-})\}dt+a(Y_{t-})dB_t\notag\\
&+\int_{|\xi|<1}[V(Y_{t-}+\xi)-V(Y_{t-})]\tilde{N}(dt,d\xi)+\int_{|\xi|<1}[V(Y_{t-}+\xi)-V(Y_{t-})-\frac{\xi}{c}a(Y_{t-})]\nu(d\xi)dt\notag\\
=&\{\frac{1}{c}a(Y_{t-})k(Y_{t-})+\frac{c}{2}\frac{da(x)}{dx}(Y_{t-})\}dt+a(Y_{t-})dB_t+
\int_{|\xi|<1}\frac{\xi}{c}a(Y_{t-})\tilde{N}(dt,d\xi)\notag\\
&+ \int_{|\xi|<1}[V(Y_{t-}+\xi)-V(Y_{t-})-\frac{\xi}{c}a(Y_{t-})]N(dt,d\xi)\notag\\
=&\{\frac{1}{c}a(Y_{t-})k(Y_{t-})+\frac{c}{2}\frac{da(x)}{dx}(Y_{t-})\}dt+a(Y_{t-})dB_t-
\int_{|\xi|<1}\frac{\xi}{c}a(Y_{t-})\nu(d\xi)\notag\\
&+ \int_{|\xi|<1}[V(Y_{t-}+\xi)-V(Y_{t-})]N(dt,d\xi).\notag
\end{align}
By Integrating from $s$ to $u$ and using the initial condition, we obtain
\begin{align}
\int_s^u a(Y_{t-})dB_t=&V(Y_u)-V(x_0)-\frac{1}{2}\int_s^u\{\frac{2}{c}a(Y_{t-})k(Y_{t-})+c\frac{da(x)}{dx}(Y_{t-})\}dt\notag\\
&-\int_s^u\int_{|\xi|<1}[V(Y_{t-}+\xi)-V(Y_{t-})]N(dt,d\xi)+\int_s^u\int_{|\xi|<1}\frac{\xi}{c}a(Y_{t-})\nu(d\xi)dt\notag\\
=&V(Y_u)-V(x_0)-\frac{1}{2}\int_s^u\{\frac{2}{c}a(Y_{t-})k(Y_{t-})+c\frac{da(x)}{dx}(Y_{t-})\}dt\notag\\
&-\Sigma_{s\leq t\leq u}[V(Y_t)-V(Y_{t-})]\chi_{|\xi|<1}(\Delta Y_t)+\int_s^u\int_{|\xi|<1}\frac{\xi}{c}a(Y_{t-})\nu(d\xi)dt,\notag
\end{align}
where the last step is based on the definition of Poisson random measure $N(dt, d\xi)$.\par
Replacing the Ito integral in (\ref{RND1}) by the   relation above, we get the expression $F[y(t)]$ for (\ref{RND1}) given by (\ref{F}).
\end{proof}

\begin{proof}[{\bf Proof of (\ref{F2})}]  By using It\^o formula, we get
\begin{align}
&dV_X(X_t,z_0(t))\notag\\
=&\{\frac{\partial V_X(x,z_0(t))}{\partial t}\mid_{x=X_{t-}}+\frac{1}{c}a_X(X_{t-},z_0(t))f(X_{t-})+\frac{c}{2}\frac{\partial a_X(x,z_0(t))}{\partial t}\mid_{x=X_{t-}}\}dt\notag\\
&+a_X(X_{t-},z_0)dB_t +
\int_{|\xi|<1}[V_X(X_{t-}+\xi,z_0(t))-V_X(X_{t-},z_0(t))]\tilde{N}(dt,d\xi)\notag\\
&+\int_{|\xi|<1}[V_X(X_{t-}+\xi,z_0(t))-V_X(X_{t-},z_0(t))-\frac{\xi}{c}a_X(X_{t-},z_0(t))]\nu(d\xi)dt\notag\\
=&\{\frac{\partial V_X(x,z_0(t))}{\partial t}\mid_{x=X_{t-}}+\frac{1}{c}a_X(X_{t-},z_0(t))f(X_{t-})+\frac{c}{2}\frac{\partial a_X(x,z_0(t))}{\partial t}\mid_{x=X_{t-}}\}dt\notag\\
&+a_X(X_{t-},z_0)dB_t +
\int_{|\xi|<1}\frac{\xi}{c}a_X(X_{t-},z_0(t))\tilde{N}(dt,d\xi)\notag\\
&+\int_{|\xi|<1}[V_X(X_{t-}+\xi,z_0(t))-V_X(X_{t-},z_0(t))-\frac{\xi}{c}a_X(X_{t-},z_0(t))]N(dt,d\xi)\notag\\
=&\{\frac{\partial V_X(x,z_0(t))}{\partial t}\mid_{x=X_{t-}}+\frac{1}{c}a_X(X_{t-},z_0(t))f(X_{t-})+\frac{c}{2}\frac{\partial a_X(x,z_0(t))}{\partial t}\mid_{x=X_{t-}}\}dt\notag\\
&+a_X(X_{t-},z_0)dB_t-
\int_{|\xi|<1}\frac{\xi}{c}a_X(X_{t-},z_0(t))\nu(d\xi)dt\notag\\
&+ \int_{|\xi|<1}[V_X(X_{t-}+\xi,z_0(t))-V_X(X_{t-},z_0(t))]N(dt,d\xi).\notag
\end{align}
By integrating above formula from $s$ to $u$, we obtain:
\begin{align}
&\int_s^u a_X(X_{t-},z_0)dB_t\notag\\
=&V_X(X_u,z_0(u))-V_X(x_0,z_0(s))+\int_s^u\int_{|\xi|<1}\frac{\xi}{c}a_X(X_{t-},z_0)\nu(d\xi)dt\notag\\
&-\frac{1}{2}\int_s^u\{2\frac{\partial V_X(x,z_0(t))}{\partial t}\mid_{x=X_{t-}}
+\frac{2}{c}a_X(X_{t-},z_0(t))f(X_{t-})+c\frac{\partial a_X(x,z_0(t))}{\partial t}\mid_{x=X_{t-}}\}dt\notag\\
&-\int_s^u\int_{|\xi|<1}[V_X(X_{t-}+\xi,z_0(t))-V_X(X_{t-},z_0(t))]N(dt,d\xi)\notag\\
=&V_X(X_u,z_0(u))-V_X(x_0,z_0(s))+\int_s^u\int_{|\xi|<1}\frac{\xi}{c}a_X(X_{t-},z_0)\nu(d\xi)dt\notag\\
&-\frac{1}{2}\int_s^u\{2\frac{\partial V_X(x,z_0(t))}{\partial t}\mid_{x=X_{t-}}+\frac{2}{c}a_X(X_{t-},z_0(t))f(X_{t-})
+c\frac{\partial a_X(x,z_0(t))}{\partial t}\mid_{x=X_{t-}}\}dt\notag\\
&-\Sigma_{s\leq t \leq u}[V_X(X_{t},z_0(t))-V_X(X_{t-},z_0(t))]\chi_{|\xi|<1}(\Delta X_t).\notag
\end{align}
Replacing the Ito integral in (\ref{RND2}) by the relation above, we get the expression $J_X[X_t,z_0(t)]$ for (\ref{RND2}) given by (\ref{F2}).
\end{proof}

\section*{Acknowledgements}
The authors would like to thank Professor Jianglun Wu, Dr Pingyuan Wei, Dr Wei Wei, Dr Qi Zhang, Dr Qiao Huang, Dr Yuanfei Huang, and Dr Jianyu Hu for helpful discussions. This work was partly supported by the NSF grant 1620449, and NSFC grants 11531006 and  11771449.

\end{document}